\documentclass[a4paper,11pt]{article} 

\usepackage[utf8]{inputenc} 

\usepackage{geometry} 
\usepackage{booktabs} 
\usepackage{array} 
\usepackage{paralist} 
\usepackage{amssymb}
\usepackage[section]{placeins}
\usepackage{graphicx}
\usepackage[font={small}]{caption}
\usepackage[font={footnotesize}]{subcaption}
\usepackage{algorithm}
\usepackage{algpseudocode}
\usepackage[normalem]{ulem}
\usepackage{framed}
\usepackage[toc,page]{appendix}
\usepackage{longtable}
\usepackage{pgfplotstable}
\usepackage{booktabs}
\usepackage{changepage}
\usepackage{filecontents}

\DeclareCaptionType{mycapequ}[][List of equations]
\usepackage{fancyhdr} 
\usepackage[nottoc,notlof,notlot]{tocbibind} 
\usepackage[titles,subfigure]{tocloft} 
\usepackage{amsmath,amsfonts,amsthm} 

\newtheorem{theorem}{Theorem}[section]
\newtheorem{lemma}[theorem]{Lemma}

\newtheorem{claim}[theorem]{Claim}

\newcommand{\argmin}{\operatornamewithlimits{argmin}}

\title{Bounds on Double-Sided Myopic Algorithms for Unconstrained Non-monotone Submodular Maximization}
\author{
		Norman Huang \\
	University of Toronto \\
	\texttt{huangyum@cs.toronto.edu}
	\and
Allan Borodin \\
	University of Toronto \\
	\texttt{bor@cs.toronto.edu}
}

\begin{document}
\maketitle

\begin{abstract}
Unconstrained submodular maximization captures many NP-hard combinatorial optimization problems, including Max-Cut, Max-Di-Cut, and variants of facility location problems. Recently, Buchbinder \emph{et al.} \cite{Buchbinder} presented a surprisingly simple linear time randomized \emph{greedy-like} online algorithm that achieves a constant approximation ratio of $\frac{1}{2}$, matching optimally the hardness result of Feige \emph{et al.} \cite{Feige:2007:MNS:1333875.1334170}. Motivated by the algorithm of Buchbinder \emph{et al.}, we introduce a precise algorithmic model called \emph{double-sided myopic algorithms}. We show that while the algorithm of Buchbinder \emph{et al.} can be realized as a randomized online double-sided myopic algorithm, no such deterministic algorithm, even with adaptive ordering, can achieve the same approximation ratio. With respect to the
Max-Di-Cut problem, we relate the Buchbinder \emph{et al.} algorithm and our
myopic framework to the online algorithm and inapproximation of Bar-Noy and 
Lampis \cite{barnoy:lampis}.  

\end{abstract}

\section{Introduction}

Submodularity emerges in natural settings such as economics, algorithmic game theory, and operations research; many combinatorial optimization problems can be abstracted as the maximization/minimization of submodular functions. A canonical example of such is $f(S) =  \sum_{i \in S, j \notin S} w_{ij}$, the cut of a graph with edge weights $w_{ij}$, of the vertex set $S$. As with the special case of Min-Cut, the general problem of submodular minimization is solvable in polynomial time \cite{lovasz,cunningham1985submodular,Iwata:2001:CSP:502090.502096,Schrijver:2000:CAM:361537.361552}. 

Maximizing a submodular function, on the other hand, generalizes NP-hard problems such as Max-Cut \cite{Goemans:1995:IAA:227683.227684}, Max-Di-Cut \cite{Halperin01combinatorialapproximation,alon2007maximum,Goemans:1995:IAA:227683.227684}, Maximum-Coverage \cite{Hochbaum:1996:ACP:241938.241941,Feige:1998:TLN:285055.285059}, expected influence in a social network \cite{KempeKT03} and facility location problems \cite{bankaccount,Cornuejols1977163}. Appropriately, this problem tends to be approached under the context of approximation heuristics in the literature. For monotone submodular functions, maximization under a cardinality constraint can be achieved by the greedy algorithm with an approximation ratio of $(1-\frac{1}{e})$ \cite{nemhauser}, which is in fact optimal in the \emph{value oracle model} \cite{nemhauser1978best}. The same approximation ratio is obtainable for the more general matroid constraints \cite{calinescu07maximizinga,filmus2012tight}, as well as knapsack constraints \cite{sviridenko2004note,lee2009non}.

We limit our discussion to unconstrained non-monotone submodular maximization --- typical examples of which include Max-Cut and Max-Di-Cut in graphs and hypergraphs. Note that solving these problems as general submodular maximization may not yield the best approximation ratio. In fact, Goemans and Williamson \cite{Goemans:1995:IAA:227683.227684} used semidefinite programming to approximate Max-Cut within 0.878, and Max-Di-Cut within 0.796. The approximation for Max-Di-Cut was later improved to 0.859 by Feige and Goemans \cite{feige1995approximating} and to the currently best known ratio of 0.874 by Lewin, Livnat, and Zwick \cite{lewin2006improved}. Trevisan \cite{trevisan1998parallel} showed that $\frac{1}{2}$ approximation for Max-Di-Cut can also be achieved with linear programming. On the other hand, an approximation ratio of $\frac{1}{2}$ for the general unconstrained non-monotone submodular maximization problem is optimal unless exponentially many value oracle queries are made \cite{Feige:2007:MNS:1333875.1334170,vondrak2013symmetry}. In the former case, $f$ is known to have a succinct representation (i.e. $f$ is completely revealed once the algorithm queries every edge weight); while in general, an explicit representation of $f$ might be exponential in the size of the ground set. On the other hand, instances of Max-Cut and Max-Di-Cut are convenient in establishing lower bounds for the general submodular maximization problem.

Recently, linear time (linear in counting one step per oracle call) \emph{double-sided greedy} algorithms were developed by Buchbinder \emph{et al.} \cite{Buchbinder} for unconstrained non-monotone submodular maximization. The deterministic version of their algorithm, stated formally in Algorithm~\ref{alg1}, achieves an approximation ratio of $\frac{1}{3}$, while the randomized version achieves $\frac{1}{2}$ in expectation - improving upon the $\frac{2}{5}$ randomized local-search approach in \cite{Feige:2007:MNS:1333875.1334170}, and the $0.42$ simulated-annealing technique in \cite{gharan2011submodular,feldman2011nonmonotone}, in terms of approximation ratio, time complexity and arguably, algorithmic simplicity. While the hardness result of Feige \emph{et al.} \cite{Feige:2007:MNS:1333875.1334170} implies optimality of the randomized algorithm, the gap between the deterministic and randomized variants remains an open problem. More specifically, is there any de-randomization that would preserve both the greedy aspect of the algorithm as well as the approximation?

To address this question, we adapt the framework of priority algorithms of Borodin \emph{et al.} \cite{Borodin:2002:PA:545381.545481}, a model for \emph{greedy-like} algorithms that has since been used in studying the limits of certain optimization algorithms \cite{regev2002priority}, graph algorithms \cite{davis2004models}, randomized greedy algorithms \cite{angelopoulos2004randomized}, and dynamic programming \cite{alekhnovich2005toward}. The idea is to derive information-theoretic lower bounds for entire \emph{classes} of greedy-like algorithms using adversarial arguments similar to those employed in online competitive analysis. In our case, we define a \emph{double-sided myopic algorithms} framework, and show that no such algorithm in the deterministic setting can de-randomize the Buchbinder \emph{et al.} $\frac{1}{2}$-ratio double-sided greedy algorithm.

\subsection{Related Work}

Our motivation is the double-sided greedy algorithms of Buchbinder \emph{et al.}
that we wish to formalize. Independently, Bar-Noy and Lampis 
\cite{barnoy:lampis}
gave a $\frac{1}{3}$ deterministic online greedy algorithm for the Max-Di-Cut 
problem matching the deterministic approximation obtained by Buchbinder \emph{et al.}
for all unconstrained non-monotone submodular maximization problems. 
Interestingly, their algorithm is shown to be the 
de-randomization of the simple $\frac{1}{4}$ random-cut algorithm!  
Bar-Noy and Lampis give an improved $\frac{2}{3\sqrt{3}}$ approximation for Max-Di-Cut when restricted to DAGs. Furthermore, they provide a precise online 
model with respect to which this approximation bound is essentially optimal.

In addition to the naive random-cut algorithm (and its de-randomization), there are a number of combinatorial algorithms for the Max-Di-Cut problem. Alimonti \cite{alimonti1997non} gives a $\frac{2}{5}$ non-oblivious local search algorithm. Halperin and Zwick \cite{Halperin01combinatorialapproximation} also give a linear time randomized $\frac{9}{20}$-approximation algorithm, as well as a bipartite-matching based algorithm with $\frac{1}{2}$-approximation. Based on a novel LP formulation, a deterministic algorithm due to Datar \emph{et al.} \cite{datar2003combinatorial} achieves a $\frac{4}{9}$-approximation for Max-Di-Cut. On the other hand, little is known about the limits of combinatorial approaches to Max-Di-Cut. In addition to the online inapproximation by Bar-Noy and Lampis, Feige and Jozeph \cite{feige2010oblivious} give a randomized \emph{oblivious algorithm}  
\footnote{Feige \cite{feige2010oblivious} uses the term oblivious in a different sense than Alimonti \cite{alimonti1997non} who uses the term to indicate that in each iteration of the local-search, the algorithm uses an auxiliary function rather than the given objective function.} (where only the total in/out weights of a vertex is given) that achieves a $0.483$ ratio, but show that no oblivious algorithm can do better than $0.4899$. In Section~\ref{sec:discussion}, 
we discuss the relation of oblivious algorithms to our work. Finally, we note that beyond combinatorial algorithms, H{\aa}stad \cite{haastad2001some} gives a $\frac{11}{12}$ inapproximability bound for Max-Di-Cut under the assumption that $NP \neq P$; whereas for Max-Cut, Khot \emph{et al.} \cite{khot2007optimal} show that the ratio of $0.878$ is optimal under the Unique Games Conjecture.

Another relevant class of problems is submodular Max-Sat (in which the objective function is monotone submodular). Azar \emph{et al.} \cite{azar2011submodular} provide a randomized online algorithm that achieves $\frac{2}{3}$-approximation, which is tight under their data model. They also demonstrate via a simple reduction the equivalence between submodular Max-Sat and monotone submodular maximization subject to binary partition matroids constraint. In fact, by incorporating the bipartition constraint, Buchbinder \emph{et al.} show that their $\frac{1}{2}$ randomized double-sided greedy algorithm can be readily extended to a $\frac{3}{4}$-approximation algorithm for submodular Max-Sat. While the $\frac{3}{4}$ ratio was already achievable by randomized algorithms \cite{poloczek2011randomized,van2012simpler}, the double-sided algorithm generalizes to submodular Max-Sat, all the while entailing a much simpler analysis (see Poloczek \emph{et al.} \cite{poloczek2013some} for details). On the other hand, Poloczek \cite{poloczek2011bounds} shows that for Max-Sat, no deterministic adaptive priority algorithm 
(in a more general input model than in Azar \emph{et al.}) can achieve an approximation ratio of $0.729$. This rules out the possibility of de-randomizing the above $\frac{3}{4}$ algorithms using any greedy algorithm.\footnote{The inapproximation result of Poloczek also applies to submodular Max-Sat, as submodularity encompasses all modular (linear) functions.}
This stands in contrast to the fact that the naive randomized algorithm can 
be de-randomized  
(by the method of conditional expectations), and as shown by 
Yannakakis \cite{Yannakakis94} becomes   
Johnson's \cite{Johnson74} deterministic algorithm. Chen \emph{et al.} \cite{ChenFZ99} show that 
Johnson's algorithm is a $\frac{2}{3}$ approximation for Max-Sat. 


Independent of our work, Paul, Poloczek and Williamson \cite{PaulPW2014}
have very recently derived a number of deterministic algorithms, and deterministic and randomized inapproximations for Max-Di-Cut with respect to the priority algorithm framework.

\subsection{Basic Definitions}
A set function $f: 2^\mathcal{N} \to \mathbb{R}$ is \emph{submodular} if for any $S,T \subseteq \mathcal{N}$,

\begin{align*}
	f(S\cup T) + f(S \cap T) \leq f(S) + f(T).
\end{align*}

We say that $f$ is \emph{monotone} if $f(S) \leq f(T)$ for all $S\subseteq T \subseteq \mathcal{N}$, and \emph{non-monotone} otherwise. An equivalent, and perhaps more intuitive definition of submodular functions captures the principle of \emph{diminishing returns}: $f(S \cup \{x\}) - f(S) \geq f(T \cup \{x\}) - f(T)$, whenever $S \subseteq T$ and $x \in \mathcal{N} \setminus T$. 

In the unconstrained non-monotone submodular maximization problem, we are given 
a finite subset $\mathcal{N}$ and the goal is 
to find a subset $S \subseteq \mathcal{N}$ (where $\mathcal{N}$ is finite) so as to maximize $f(S)$ 
for a specified submodular function $f$.
In general, since the specification of $f$ requires knowing 
its value on all possible subsets, $f$ is accessed via a value oracle which 
given 
$X \subseteq \mathcal{N}$, returns $f(X)$. We state the deterministic version of the Buchbinder \emph{et al.} double greedy algorithm, which approximates this problem with $\frac{1}{3}$ guarantee, in Algorithm~\ref{alg1}. 

\begin{algorithm}[htpb]
 \begin{algorithmic}[1]
 \State $S_0 \gets \emptyset, T_0 \gets \mathcal{N}$
 \For{$i = 1 \text{ to } n$}
 \State $a_i\gets f(S_{i-1} \cup \{u_i\}) - f(S_{i-1})$
 \State $b_i\gets f(T_{i-1} \setminus \{u_i\}) - f(T_{i-1})$
 \If{$a_i \geq b_i$} \State $S_i \gets S_{i-1} \cup \{u_i\}, T_i \gets T_{i-1}$ 
 \Else \State $T_i \gets T_{i-1} \setminus \{u_i\}, S_i \gets S_{i-1}$ \EndIf
 \EndFor
 \State \textbf{return} $S_n$
 \end{algorithmic}
 \caption{DeterministicUSM($f,\mathcal{N}$)}
 \label{alg1}
\end{algorithm}

For some explicitly defined submodular functions, such as Max-Cut and Max-Di-Cut, 
we are given as input an edge weighted graph 
(or directed graph) $G = (V,E,w)$. 
Here we interpret the ground set $\mathcal{N}$ to be vertex set $V$, and $f(X) = \sum_{u \in X, v \in V \setminus X} w(u,v)$ to be the cut function, which can be assumed to
be computed at unit cost. But, of course, in such an explicitly defined
problem, if an algorithm is given the edge weights then it may deduce the complete mapping of $f$
without value oracle calls. For online computations, we usually
assume that the graph is revealed one vertex at a time, in the sense
that the revealed vertex specifies the edges and their weights to
all previously revealed vertices (while the number of vertices may be given a priori). As argued, by Bar-Noy and Lampis, 
for cut problems, the revealed vertex must also give global information about
each node, namely the total weight of in-edges and the 
total weight of out-edges. 

\subsection{Our Contribution}  

We introduce a formalization of \emph{double-sided myopic algorithms} - an adaptation of the priority framework under a \emph{restricted} value oracle model. To make this precise, we will introduce three types of \emph{relevant oracle queries}. We show that the double-sided greedy algorithms of Buchbinder \emph{et al.} can be realized as online double-sided myopic algorithms. Moreover, our framework
also captures the online algorithm of Bar-Noy and Lampis for Max-Di-Cut. 
Even our most
restrictive model (see query Q-Type 1 in Section~\ref{queryclasses}) 
allows for plausible ways to extend these algorithms. As our main contribution, we establish the following lower bounds for deterministic algorithms under this framework with respect to the stronger Q-Types as defined in Section~\ref{queryclasses}.

\begin{theorem}
\label{onlinetheorem} 
No deterministic online double-sided myopic algorithm (with respect to our
strongest oracle model Q-Type 3) can achieve a competitive ratio of $\frac{2}{3\sqrt{3}} + \epsilon \approx 0.385 + \epsilon$ for any $\epsilon > 0$ for the unconstrained non-monotone submodular maximization problem.
\end{theorem}

Theorem~\ref{onlinetheorem} is obtained by directly applying the hardness result of Bar-Noy and Lampis in Max-Di-Cut \cite{barnoy:lampis}, which, to the best of our knowledge, has not been studied in the context of the Buchbinder \emph{et al.} double greedy algorithm. We also show that the deterministic algorithm with $\frac{1}{3}$ approximation ratio guarantee for Max-Di-Cut by Bar-Noy and Lampis is in fact an instantiation of the double-sided greedy algorithm when the submodular function $f$ is the directed cut function. 

Extending to the class of (deterministic) fixed and adaptive priority algorithms in Theorem~\ref{fixedtheorem} and \ref{adaptivetheorem}, we construct submodular functions and corresponding adversarial strategies that would force an inapproximability ratio strictly less than $\frac{1}{2}$. In terms of oracle restrictiveness, our inapproximation holds for the \emph{already attained partial solution query} model (Q-Type 2), which is more powerful than what is sufficient to achieve the $\frac{1}{2}$ ratio by the online randomized greedy algorithm (Q-Type 1). A comprehensive description of the different oracle models will be covered in Section~\ref{queryclasses}. Our proof is computer assisted, in that the objective function value for each possible subset is explicitly computed through linear programming.

\begin{theorem}
\label{fixedtheorem}
 There exists a problem instance such that no fixed priority double-sided myopic algorithm with respect to oracle model Q-Type 2 can achieve an approximation ratio better than 0.428.
\end{theorem}

\begin{theorem}
\label{fixedtheoremtype3}
 There exists a problem instance such that no fixed priority double-sided myopic algorithm with respect to oracle model Q-Type 3 can achieve an approximation ratio better than 0.450.
\end{theorem}

\begin{theorem}
\label{adaptivetheorem}
 There exists a problem instance such that no adaptive priority double-sided myopic algorithm with respect to oracle model Q-Type 2 can achieve an approximation ratio better than 0.432.
\end{theorem}

\section{The Double-Sided Myopic Algorithms Framework}
\label{framework}

By integrating a \emph{restricted} value oracle model in a priority framework, we propose a general class of \emph{double-sided myopic algorithms} that captures the Buchbinder \emph{et al.} double-sided greedy algorithm. We rephrase the double-sided procedure as a single-sided sweep, but with access to a pair of complementary objective functions - as opposed to simultaneously evolving a bottom-up and a top-down solution. This interpretation admits a myopic behavior common to most greedy algorithms; and we show how this facilitates an adaptation of a priority-like framework. To make this formalization precise, we must specify 
how information about input items is represented and accessed. The generality of the unconstrained submodular maximization problem raises some representational issues leading to a number of different input presentations. To address these issues, we employ a \emph{marginal value} representation that is compatible with both the value oracle model and the priority framework.

On the one hand, a precise description of a data item is necessary in determining an ordering of the input set, as well as in quantifying the availability of information in the decision step. On the other hand, the value oracle model 
measures complexity in terms of information access. An apparent incompatibility arises when an exact description of a data item can trivialize query complexity - as in the case of the \emph{marginal value} representation, where exponentially many queries are needed to fully describe an item when $f$ is an arbitrary submodular function. For this reason, we propose a hierarchy of oracle restrictions that categorizes the concept of myopic \emph{short-sightedness}, while preserving the fundamental characteristics of the priority framework. This consequently gives rise to a hierarchy of algorithmic models, which will be actualized once we define the algorithm's \emph{internal memory}. We conclude this section by expressing the Buchbinder \emph{et al.} double-sided greedy algorithm under the double-sided myopic framework.

\subsection{Value Oracle and the Marginal Value Representation}
\label{sec:oracle-models}

For succinctly encoded problems, the objective is a function of an explicitly given collection of input item attributes; that is, weighted adjacencies in Max-Cut and Max-Di-Cut, and distances and opening costs in facility location problems. In contrast, in general submodular maximization problems, other than labels, the elements of $\mathcal{N}$ by themselves do not convey pertinent information. 
Therefore we interpret $\mathcal{N}$ as a fixed ground set, and the input instance as the objective function itself, drawn from the family $\mathcal{F} = \{f | f:2^\mathcal{N} \to \mathbb{R} \}$ of all submodular set functions over $\mathcal{N}$. To avoid having an exponentially sized input, we employ a \emph{value oracle} as an intermediary between the algorithm and the input function. That is, given a query $S \subseteq \mathcal{N}$, the value oracle will answer the question: ``What is the value of $f(S)$?" 
By an abuse of notation, we will interchangeably refer to $f$ as the objective function (\emph{i.e.} when referring to $f(S)$ as a real number) and as the value oracle (\emph{i.e.} when an algorithm submits a query to $f$). 

We also introduce for notational convenience a complementary oracle $\bar{f}$, such that $\bar{f}(X) = f(\mathcal{N} \setminus X)$, for input set $\mathcal{N}$. 
This allows us to express the double-sided myopic algorithm similar to the priority setting in \cite{Borodin:2002:PA:545381.545481}, where the solution set $X$ is constructed item by item, using only locally available information.
In other words, the introduction of $\bar{f}$ allows access to $f(\mathcal{N} \setminus X)$ using $X$ (which is composed of items that have already been considered) as query argument, instead of $\mathcal{N} \setminus X$.
The {\bf double-sidedness} of the framework follows in the sense that $f$ and $\bar{f}$ can be simultaneously accessed.

We wish to model \emph{greedy-like} algorithms that process the problem instance item by item. 
But what is an item when considering an arbitrary submodular function? While the natural choice is that an item is an element of ${\cal N}$ (to include or
not include in the solution $S$), for arbitrary submodular maximization the input is a function $f$ (or more precisely, an oracle interface of $f$) and thus the notion of a \emph{data item} becomes somewhat problematic. To address this issue, we propose the \emph{marginal value representation}, in which $f$ is instilled into the elements of $\mathcal{N}$. Specifically, we could describe a data item as an element $u \in \mathcal{N}$, plus a list of marginal differences $\rho(u|S) = f(S \cup \{u\}) - f(S)$, and $\bar{\rho}(u|S) = \bar{f}(S \cup \{u\}) - \bar{f}(S)$ for every subset $S \subseteq \mathcal{N}$. 
The impracticality in using a complete representation in this form is evident as the space required is exponential in $|\mathcal{N}|$. Furthermore,
such a complete marginal representation would lead to a trivial optimal greedy  
algorithm by adaptively choosing the next item to be one 
that is included in an optimal solution (given what has already been accepted). Therefore we assume that 
an oracle query must be made by the algorithm in order to access each marginal value. 
In terms of inapproximability arguments, if we impose certain constraints on what oracle queries are allowed during the computation, then
the model allows us to justify when input items are indistinguishable.
In fact, our inapproximation results will not be based on bounding
of the number of oracle queries but rather will  be imposed by the restricted
myopic nature of the algorithm. 

\subsection{Classes of Relevant Oracle Queries}
\label{queryclasses}

The {\bf myopic} condition of our framework is imposed both by the nature of the ordering in the priority model, which we describe later on, as well as by restricting the algorithm to make only certain types of oracle queries that are \emph{relevant}. 
To avoid ambiguity, we emphasize that the value oracles are given in terms of $f$ and $\bar{f}$, with $\rho$ and $\bar{\rho}$ being used purely for notational simplicity. That is, when we say that the algorithm learns the value $\rho(u | S)$, we assume that it queries both $f(S \cup \{u\})$ and $f(S)$.
At iteration $i$, define (respectively) $X_{i-1}$ and $Y_{i-1}$ to be the currently accepted and rejected sets, and $u_{i}$ to be the next item considered by the algorithm. We now introduce three models of relevant oracle queries in hierarchical ordering, starting with the most restrictive: \\

\begin{itemize}
\item[] \uline{Next attainable partial solution query} ({\bf Q-Type 1})

The algorithm is permitted to only query $f(X_{i-1} \cup \{u_i\})$ and $\bar{f}(Y_{i-1} \cup \{u_i\})$, corresponding to the values of the next possible partial solution. In terms of the data item model, this is equivalent to learning $\rho(u_i | X_{i-1})$ and $\bar\rho(u_i | Y_{i-1})$ for the item $u_i$. \\
We note that in this and in all our models, we can then use this information
in any way. For example, the deterministic algorithm of  
Buchbinder \emph{et al.} greedily chooses 
to add $u_i$ to $X_{i-1}$ if 
$\rho(u_i | X_{i-1}) = a \geq b = \bar\rho(u_i | Y_{i-1})$. In our model, the decision
about $u_i$ can be any (even non-computable) function of $a$ and $b$ and the
``history''of the algorithm thus far.

\item[] \uline{Already attained partial solutions query} ({\bf Q-Type 2})

In this model we allow queries of the form $\rho(u_i | X_j)$ and $\bar\rho(u_i | Y_j)$ for all $j<i$.

\item[] \uline{All subsets query} ({\bf Q-Type 3})

By making an all subsets query on $u_i$, the algorithm learns the marginal gains $\rho(u_i | S)$ and $\bar\rho(u_i | S)$ for every $S \subseteq X_{i-1} \cup Y_{i-1}$. Essentially, the algorithm is given full disclosure of the submodular function $f$ and $\bar{f}$ over the set of currently revealed items. 
Note that in this very general model, the algorithm can potentially
query exponentially many sets so that in principle such algorithms are not
subject to the $\frac{1}{2}$ inapproximation result of Feige et al
\cite{feige2010oblivious}.

\end{itemize}

Finally, observe that the above classes are ordered by inclusion. That is, if $Q^{i}(\mathcal{N})$ is the set of all queries on input $\mathcal{N}$ permitted under query type $i$, then $Q^{i}(\mathcal{N}) \subseteq Q^{i+1}(\mathcal{N})$. 

\subsection{Internal Memory or History}

We define an algorithm's \emph{internal memory} or 
\emph{history}  as a record of the following:

\begin{itemize}
\item All previously considered items and the decisions made for these items.
\item 

The outcomes of all previous relevant query results.

\item
Anything that can be deduced from all previously considered items, decisions and 
relevant queries. That is, in the priority framework (section \ref{sec:priority-models}), the order in which the items are considered will determine that certain items in $\mathcal{N}$ \emph{cannot} take on certain marginal values. In other words, the algorithm may rule out from $\mathcal{F}$ all submodular functions that would contradict the observed ordering. 
\end{itemize}


Recalling the definition of $X_{i-1}$ and $Y_{i-1}$, let $\mathcal{N}_{i-1} = X_{i-1} \cup Y_{i-1}$ be the set of all previously seen (and decided upon) items at the start of iteration $i$. Let $u_i$ be the next element, and allow the algorithm to perform any possible relevant queries. We summarize the algorithm's internal memory under each query restriction type: 

\begin{itemize}

\item[] \uline{Q-Type 1 myopic model} 

An algorithm with Q-Type 1 oracle access has the record of:

\begin{itemize}
\item The decision made for every $u \in \mathcal{N}_{i-1}$. 
\item $\rho(u_j | X_{j-1})$ and $\bar\rho(u_j | Y_{j-1})$ for $0 < j \leq i$. 
\end{itemize}

\item[] \uline{Q-Type 2 myopic model}

An algorithm with Q-Type 2 oracle access has the record of:
\begin{itemize}
\item The decision made for every $u \in \mathcal{N}_{i-1}$.
\item $\rho(u_j | X_k)$ and $\bar\rho(u_j | Y_k)$ for $0 \leq k < j \leq i$.
\end{itemize}

\item[] \uline{Q-Type 3 myopic model} 

An algorithm with Q-Type 3 oracle access has the record of:

\begin{itemize}
\item The decision made for every $u \in \mathcal{N}_{i-1}$.
\item $\rho(u | S)$ and $\bar\rho(u | S)$ for all $u \in \mathcal{N}_i$ and all $S \subseteq \mathcal{N}_i$.
\end{itemize}

\end{itemize} 



\subsection{Priority Models}
\label{sec:priority-models}

The other aspect of the priority framework that we need to specify is
the order in which input items are considered. 
We present a high-level description of the generic templates of double-sided myopic algorithms.  As in the priority framework, we categorize double-sided myopic algorithms into the following subclasses: {\bf online}, {\bf fixed} priority, and {\bf adaptive} priority. For all templates, the decision step
remains the same. Namely, based on the history (which depends on
the particular value oracle model) of previous relevant queries,
and
relevant queries corresponding to  current item being considered, the
algorithm makes an irrevocable decision for the current item. 
Let $\mathcal{N}$ be the input set, whose length $n$ is the only information that is initially accessible (\emph{i.e.} before any queries are made) to the algorithm. We note that the transparency of the input length $n$ allows us to capture a broader and potentially more powerful class of algorithms compared to that of the original priority framework. 
Or more precisely, it prevents an adversary from abruptly ending a computation freeing an algorithm from this concern. 

An online double-sided myopic algorithm conforms to the standard template of an online algorithm: 

\begin{framed}
		
                \noindent \uline{Online 2-Sided Myopic Algorithm}
		\begin{itemize}[]
		\item while not empty($\mathcal{N}$)
                \end{itemize}
		\begin{adjustwidth}{2.3em}{0pt} \begin{itemize}[]
			\item \emph{next} := lowest index (determined by adversary) of items remaining in $\mathcal{N}$ 
\item {\bf Relevant Query:} 
Perform any set of relevant queries and update the \emph{internal memory}
\item {\bf Decision:} As a function of the \emph{internal memory}, irrevocably accept or reject $u_{next}$, and remove $u_{next}$ from $\mathcal{N}$
\end{itemize}
\end{adjustwidth}
        
\end{framed}

A fixed priority algorithm has some limited ability to determine the ordering of 
the input items. Namely, such an algorithm specifies an injective priority function $\pi : \mathbb{R} \times \mathbb{R} \to \mathbb{R}$, such that the item that minimizes $\pi$ corresponds to the item of highest priority. 
In the generality of submodular function maximization, lacking any other information, the priority of an input $u$ is determined as a function of $u$'s marginal gains in $f$ and $\bar{f}$ with respect to the empty set. The item $u_{next} \in \mathcal{N}$ which minimizes $\pi$ is then given to the algorithm. 
Our inapproximations also
hold if the algorithm is also (say) given the maximum value of any $f(\{u\})$. 
In such a case, the priority can be a more complex function of the precise values of these
marginals. We emphasize that $\pi$ is determined before the algorithm makes any oracle queries, and cannot be changed. 
The structure of a fixed priority algorithm is as follows: 

\begin{framed}
		
                \noindent \uline{Fixed Priority 2-Sided Myopic Algorithm} 
		\begin{itemize}[]

		 \item {\bf Ordering:} Specify a priority function $\pi: \mathbb{R}\times \mathbb{R} \to \mathbb{R}$

		\item while not empty($\mathcal{N}$) 
\end{itemize}
		\begin{adjustwidth}{2.3em}{0pt} \begin{itemize}[]
			  \item\emph{next} := index $i$ of the item in $\mathcal{N}$ that minimizes $\pi(\rho(u_i | \emptyset), \bar\rho(u_i | \emptyset))$ 
\item {\bf Relevant Query:} Using $u_{next}$ as the next item, perform any set of relevant queries and update the \emph{internal memory}
\item {\bf Decision:} As a function of the \emph{internal memory}, irrevocably accept or reject $u_{next}$, remove $u_{next}$ from $\mathcal{N}$ and update internal memory.
		\end{itemize}
\end{adjustwidth}
        
\end{framed}

In the most general class of adaptive priority algorithms, a new ordering function may be specified after each item is processed. 
Here we calculate the priority of an item $u$ by extending $\pi$ to take as input all marginal differences for $u$ permissible in the chosen query model. 
More precisely, the algorithm does not inquire the marginals for all items currently in $\mathcal{N}$, but is simply given the item of highest priority $u$, namely $u = \argmin_v[\pi(Q(v))]$. 
We define $Q(u)$ to be the vector of $u$'s marginal gains accessible under the appropriate relevant query model. Observe that in Q-type 2 and 3, the length of $Q$ increases with the iterations.


\begin{framed}
		
                \noindent \uline{Adaptive Priority 2-Sided Myopic Algorithm} 
		\begin{itemize}[]
		\item while not empty($\mathcal{N}$) 
		\end{itemize}
			\begin{adjustwidth}{2.3em}{0pt} 
			\begin{itemize}[]
			\item{\bf Ordering:} Based on the \emph{internal memory}, specify a priority function $\pi$
			 \item \emph{next} := index $i$ of the item in $\mathcal{N}$ that minimizes $\pi(Q(u_i))$ 
			\item {\bf Relevant Query:} Using $u_{next}$ as the next item, perform any set of relevant queries and update the \emph{internal memory}
			\item {\bf Decision:} As a function of the \emph{internal memory}, irrevocably accept or reject $u_{next}$, remove $u_{next}$ from $\mathcal{N}$ and update internal memory.
			\end{itemize}
			\end{adjustwidth}
        
\end{framed}

We remind the reader that for the fixed and adaptive priority models, updating the internal memory also means deducing that certain marginal descriptions cannot exist and applying relevant queries (for the given $Q$-type) to obtain additional information.
In particular, if $u_{next} \in \mathcal{N}$ is the item with the minimum $\pi$ value, then any item $u_j$ 
with $\pi(Q(u_j)) < \pi(Q(u_{next}))$ cannot appear later.
Knowledge of what items (in terms of marginal representation)
cannot be in $\mathcal{N}$ can be assumed to
be part of the internal memory.  

Finally, we note that all conditions mentioned here are devoid of complexity assumptions, and thus our inapproximability arguments result from information theoretic arguments. That is, the complexity or even computability
of the ordering and decision steps is arbitrary, and there is no limitation on the size of the memory.
 
\subsection{Remodeling the Buchbinder \emph{et al.} Double-Sided Greedy Algorithm}

\begin{claim}
\label{reduction}
 The deterministic double-sided greedy algorithm in \cite{Buchbinder} can be modeled by Algorithm~\ref{alg2}, a Q-Type 1 online double-sided myopic algorithm. \end{claim}
\begin{proof}

To prove this claim, we relabel the variables and function calls on the syntactic level without modifying the algorithmic behavior, until we obtain a description that complies with the online myopic model. First we recall the formal description of the deterministic double-sided greedy algorithm (called \emph{DeterministicUSM}) in Algorithm~\ref{alg1}.

Define the variables $X_i = S_i$, $Y_i = \mathcal{N} \setminus T_i$, and rewrite Algorithm~\ref{alg1} in terms of $X_i$ and $Y_i$. 
To establish an online setting, we assume a predetermined ordering over $\mathcal{N}$. Furthermore, as argued in section~\ref{sec:oracle-models}, we will replace $\mathcal{N}$ by $\bar{f}$ as input parameter in order to demonstrate the algorithm's item-by-item behavior.
This is possible in Line 1, since $X_0 = Y_0 = \emptyset$ by definition. In the later steps, recovering $S_i$ from $X_i$ is trivial, but $T_i = \mathcal{N} \setminus Y_i$ would require access to $\mathcal{N}$. To resolve this problem, we make use of the complementary value oracle $\bar{f}$, with the following equalities

\begin{align}
	T_{i-1} \setminus \{u_i\} &= (\mathcal{N} \setminus Y_i) \setminus \{u_i\} = \mathcal{N} \setminus (Y_i \cup \{u_i\}) \\
	f(T_{i-1} \setminus \{u_i\}) - f(T_{i-1}) &= f(\mathcal{N} \setminus  (Y_i \cup \{u_i\})) - f(\mathcal{N} \setminus Y_i) \notag \\
	&= \bar{f}(Y_i \cup \{u_i\}) - \bar{f}(Y_i)
\end{align}

In this sense, we can replace $\mathcal{N}$ by $\bar{f}$ in the input parameter, since the necessary information about $\mathcal{N}$ is implicitly encoded in the function $\bar{f}$. Due to the absence of a predefined item ordering, we assume the items are labeled by an adversary. We emphasize that the ability to conceal $n$ is consequential only in the online model, and that our lower bounds for fixed and adaptive order algorithms hold even if $n$ is revealed to the algorithm. 

\begin{algorithm}[htpb]
 \begin{algorithmic}[1]
 \State $X_0 \gets \emptyset, Y_0 \gets \emptyset$
 \For{$i = 1 \text{ to } n$}
 \State $a_i\gets f(X_{i-1} \cup \{u_{i}\}) - f(X_{i-1})$
 \State $b_i\gets \bar{f}(Y_{i-1} \cup \{u_{i}\}) - \bar{f}(Y_{i-1})$
 \If{$a_i \geq b_i$} \State $X_i \gets X_{i-1} \cup \{u_{i}\}, Y_i \gets Y_{i-1}$ 
 \Else \State $Y_i \gets Y_{i-1} \cup \{u_{i}\}, X_i \gets X_{i-1}$ \EndIf
 \EndFor
 \State \textbf{return} $X_{n}$
 \end{algorithmic}
 \caption{OnlineMyopic($f,\bar{f}$)}
 \label{alg2}
\end{algorithm}

Combining these ideas, we describe the double-sided greedy algorithm as an online double-sided myopic algorithm in Algorithm~\ref{alg2} (\emph{OnlineMyopic}). Finally, observe that the value oracle access in Line 3 and 4 conform to the specifications imposed by Q-Type 1.

\end{proof}

\begin{claim} 
The randomized double-sided greedy algorithm in \cite{Buchbinder} can be modeled by a random-choice Q-Type 1 online double-sided myopic algorithm. \end{claim}
\begin{proof}
	Using the same reduction from Claim~\ref{reduction}, the proof follows identically.
\end{proof}

\section{A $\frac{2}{3\sqrt{3}}$ Inapproximation for the Online Case}
\label{online lowerbound}

Given a digraph $G(V,E)$ with non-negatively weighted edges $w(e)$, the cut value of a subset $V' \subseteq V$ is defined to be the weight of the edge set $C = \{ (u,v) | u \in V', v \in V\setminus V' \}$. In Max-Di-Cut, the objective is to find the vertex set that maximizes $\sum_{e \in C} w(e)$. Let $f : V \to \mathbb{R}^+$ be the cut value function, it can be easily verified that $f$ is non-monotone submodular. Therefore, to prove an inapproximability result 
(within some algorithmic model) for the general non-monotone submodular maximization problem, it suffices to demonstrate the existence of a hard instance of Max-Di-Cut that forces a bad approximation ratio for the algorithmic model under consideration. In particular, we utilize the online model and results of Bar-Noy and Lampis \footnote{To the best of our knowledge, the Buchbinder \emph{et al.} \cite{Buchbinder} and Bar-Noy and Lampis \cite{barnoy:lampis} papers were independent of each other.} \cite{barnoy:lampis}.

In the online model of \cite{barnoy:lampis}, 
when an input item $v$ is revealed, the algorithm is given access to the following information:

\begin{itemize}
\item $w_{in}(v)$, the total weight of incoming edges to $v$
\item $w_{out}(v)$, the total weight of outgoing edges from $v$
\item The weights of both in and out edges connecting $v$ to previously revealed vertices
\end{itemize}

Bar-Noy and Lampis prove the following theorem for their online model:

\begin{theorem} \cite{barnoy:lampis} No deterministic algorithm 
(within their data model) can achieve a competitive ratio of $\frac{2}{3\sqrt{3}}+\epsilon  \approx 0.385 + \epsilon$ for any $\epsilon > 0$ for the online Max-Di-Cut problem on DAGs. \footnote{When restricted to DAGs, Bar-Noy and Lampis also provide an algorithm for Max-Di-Cut with matching approximation ratio.}
\label{barnoybound}
\end{theorem}

In order to generalize Theorem~\ref{barnoybound} to non-monotone submodular maximization, we need to show that the data model in \cite{barnoy:lampis} is powerful enough to simulate relevant value oracle queries. 
%
%

\noindent Let $X_{i-1}$ and $Y_{i-1}$ be the set of vertices accepted and rejected, respectively, by the algorithm when $v_i$ is revealed. Notice that from the third item in the above list, the algorithm can calculate the cut between $v_i$ and $S$, for any $S \subseteq X_{i-1} \cup Y_{i-1}$. 

\begin{claim} For Max-Di-Cut, any information obtainable through Q-Type 3 oracle queries can be inferred in the data model of \cite{barnoy:lampis}. The converse, however, does not hold.
\label{info}
\end{claim}
\begin{proof}
Recall that the allowable queries under Q-Type 3 are $\rho(u|S)$, and $\bar\rho(u|S)$ for any $S \subseteq \mathcal{N}_i$. It suffices to show that these values can be calculated using the Max-Di-Cut data model. Given that $f$ is the directed cut function, $f(S)$ (\emph{resp.} $\bar{f}(S)$) can be expressed as the sum of outgoing (\emph{resp.} incoming) edge weights of vertex set $S$,

\begin{align*}
	f(S) = \sum_{s \in S} w_{out}(s) - \sum_{s,t \in S} w(s,t) \\
	\bar{f}(S) = \sum_{s \in S} w_{in}(s) - \sum_{s,t \in S} w(s,t)
\end{align*}

\noindent Then the marginal difference for a vertex $v \notin S$ with respect to $S$ is

\begin{align}
	f(S\cup \{v\}) - f(S) &= w_{out}(v) - c(S, v) - c(v,S) \label{one} \\
	\bar{f}(S \cup \{v\}) - \bar{f}(S) &= w_{in}(v) - c(v,S) - c(S,v) \label{two}
\end{align}

\noindent Since all elements of $S$ have been revealed to the algorithm previously, the data description of $v$ includes the complete edge weight information between $v$ and every vertex of $S$. Thus the cut (in either direction) between $v$ and $S$ can be directly computed. \\

\begin{figure}[htpb]
  \centering
     \includegraphics[width=0.5\textwidth]{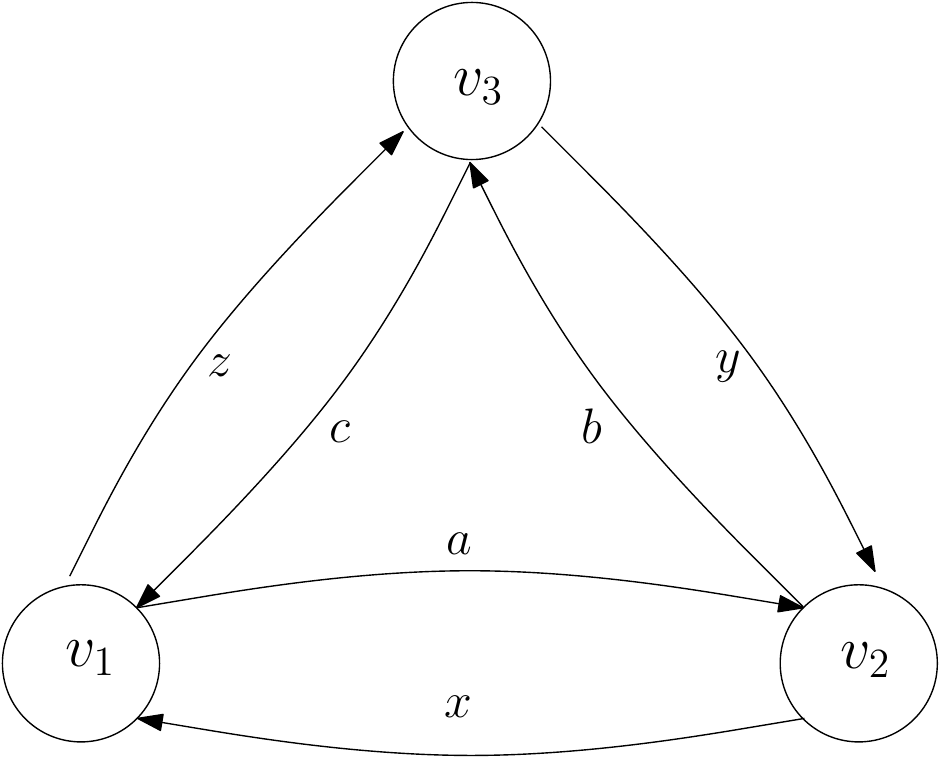}
	\caption{Directed 3 cycle with 6 edge weights.}
	\label{3cycle}
\end{figure}

For the converse, consider the directed 3-cycle $G = (V,E)$ depicted in Figure~\ref{3cycle} as counterexample. We show that even with complete disclosure of the cut function $f : 2^V \to \mathbb{R}$, we cannot reconstruct the edge weights $w : E \to \mathbb{R}$. Since $f(\emptyset) = f(V) = 0$ is trivial, there are $2^{|V|} - 2 = 6$ available linear equations. In comparison, there are $|E|=6$ unknown edge weights. Thus, with equal number of variables and linear equations,  a non-trivial solution exist if and only if the determinant is non-zero. We enumerate all subsets and their corresponding cut values in terms of the edge weights:

\begin{align*}
	c_1 = f(\{v_1\}) = a + z \\
	c_2 = f(\{v_2\}) = b + x \\
	c_3 = f(\{v_3\}) = c + y \\
	c_4 = f(\{v_1,v_2\}) = b + z\\
	c_5 = f(\{v_2,v_3\}) = c + x \\
	c_6 = f(\{v_3,v_1\}) = a + y\\
\end{align*}

It's easy to see that $c_1+c_2+c_3 = c_4 + c_5 + c_6$, implying linear dependency in the above system. Therefore, a unique solution can not be determined from the cut values alone. The argument can be extended to any input instance of size $n>3$, by constructing a graph that contains a disjoint 3-cycle.  

\end{proof}

\begin{claim}
\label{revealn}
The inapproximation of Theorem~\ref{barnoybound} holds even if the number of vertices is initially revealed to the algorithm.
\end{claim}
\begin{proof}
We prove this claim using a simple padding strategy. Define $\mathcal{A}$ to be the adversary used in Theorem~\ref{barnoybound}, we will construct a modified adversary $\mathcal{A'}$ against algorithms that knows the length of the input. Let $n$ be the size of the largest graph that $\mathcal{A}$ might present, then $\mathcal{A'}$ announces $n$ as the input size to the algorithm. $\mathcal{A'}$ will simply copy the actions of $\mathcal{A}$. If $\mathcal{A}$ terminates at iteration $k \leq n$, then $\mathcal{A'}$ sets $v_{k+1},...,v_{n}$ as isolated vertices. Since an isolated vertex $v_{iso}$ has no adjacent edges, $f(S \cup  \{v_{iso}\}) = f(S)$ for any set $S$. In the last $n-k$ iterations, the value of the algorithm's solution is therefore unchanged, regardless of its decision. On the other hand, using the same optimal solution as $\mathcal{A}$ will result in the desired approximation ratio. 
\end{proof}

We now prove Theorem~\ref{onlinetheorem}:

\begin{proof}
[Proof of Theorem~\ref{onlinetheorem}]
Assume the contrary, then by Claim~\ref{info} and \ref{revealn} we can simulate such an algorithm with one under the online Max-Di-Cut model when $f$ is a cut function on a DAG. This would contradict Theorem~\ref{barnoybound}.
\end{proof}

\subsection{Reinterpreting the Online Algorithm of Bar-Noy and Lampis}

In this section we re-examine the doubling online algorithm in \cite{barnoy:lampis} for the Max-Di-Cut problem on general graphs and compare it to the deterministic version of the double-sided greedy algorithm by Buchbinder \emph{et al.}.

\begin{claim} The doubling online algorithm of \cite{barnoy:lampis} is an instantiation of the deterministic double-sided greedy algorithm of \cite{Buchbinder}.
\end{claim}

\begin{proof}

The doubling online algorithm for Max-Di-Cut can be described simply by the following comparison:

\begin{align}
\label{ineqbarnoy}
C_0(v_{i+1}) + \frac{P_0(v_{i+1})}{2} \stackrel{\text{ ?}}{<} C_1(v_{i+1}) + \frac{P_1(v_{i+1})}{2}
\end{align}

\noindent where $C_0(v)$ (\emph{resp.} $C_1(v)$) is the certain payoff of accepting (\emph{resp.} rejecting) $v$, while $P_0(v)$ (\emph{resp.} $P_1(v)$) is the potential payoff of accepting (\emph{resp.} rejecting) $v$. Then vertex $v_{i+1}$ is accepted if Inequality~\ref{ineqbarnoy} holds, and rejected otherwise. On the other hand, the double-sided greedy algorithm of Buchbinder \emph{et al.} compares the incremental gain as follows

\begin{align}
\label{ineqbuch}
f(X_i \cup \{v_{i+1}\}) - f(X_i) \stackrel{\text{ ?}}{<} \bar{f}(Y_i \cup \{v_{i+1}\}) - \bar{f}(Y_i)
\end{align}

To show that the two algorithms are in fact the same, we prove that (\ref{ineqbarnoy}) and (\ref{ineqbuch}) are equivalent. Express the terms in (\ref{ineqbarnoy}) explicitly as follows

\begin{align*}
	C_0(v_{i+1}) &= c(v_{i+1}, Y_i) \\
	C_1(v_{i+1}) &= c(X_i, v_{i+1}) \\
	P_0(v_{i+1}) &= w_{out}(v) - \big[ c(v_{i+1},Y_i) + c(v_{i+1}, X_i) \big] \\
	P_1(v_{i+1}) &= w_{in}(v) - \big[ c(Y_i, v_{i+1}) + c(X_i, v_{i+1}) \big]
\end{align*}

\noindent Then (\ref{ineqbarnoy}) becomes

\begin{align}
	c(v_{i+1}, Y_i) +& \frac{1}{2} w_{out}(v) - \frac{1}{2}  \big[ c(v_{i+1},Y_i) + c(v_{i+1}, X_i) \big]  \notag\\
	&\stackrel{\text{ ?}}{<}   c(X_i, v_{i+1}) + \frac{1}{2}w_{in}(v) - \frac{1}{2}\big[ c(Y_i, v_{i+1}) + c(X_i, v_{i+1}) \big] \notag\\
	\Leftrightarrow w_{out}(v) - c(&v_{i+1},X_i) - c(X_i, v_{i+1}) \notag\\
	&\stackrel{\text{ ?}}{<} w_{in}(v) - c(Y_i, v_{i+1}) - c(v_{i+1},Y_i)
\end{align}

\noindent Replacing the two sides of the inequality by that of Equation~\ref{one} and \ref{two} (by substituting $X_i$ and $Y_i$ as $S$, and $v_{i+1}$ as $v$) yields Inequality~\ref{ineqbuch}.

\end{proof}

As a corollary, the doubling online algorithm of Bar-Noy and Lampis can be simulated by an online double-sided myopic algorithm. Furthermore, it follows from Claim~\ref{reduction} that Q-Type 1 relevant oracle access is sufficient for this simulation. In contrast, Claim~\ref{info} indicates that the .385 online inapproximation extends to double-sided myopic algorithms under Q-Type 3 oracle constraints. 

\section{0.428 and 0.450 Inapproximations for Fixed Priority Algorithms}
\label{fixed}

Our lower bound argument for fixed priority algorithms is achieved by constructing a hard input instance using linear programming. The solution of the LP gives the complete mapping of the objective submodular function $f: 2^\mathcal{N} \to \mathbb{R}$, where $\mathcal{N}$ is a small finite ground set. In the construction of $f$, we require certain items to be indistinguishable to the algorithm, granting the adversary control over the input ordering. The drawback in this computer assisted argument is that $f$ may not have a succinct representation like a cut function or a coverage function. Therefore, we first establish the intuition behind the adversarial strategy on a small Max-Di-Cut example, albeit with a worse lower bound.

\subsection{Adversarial Strategy on Max-Di-Cut}
Consider $G=(V,E)$, a directed 6-cycle with unit edge weights, and $f$ the directed cut function on $G$. Since a fixed order priority algorithm must determine a total ordering before the input is revealed, the priority of an input item $i$ can only depend on $f(i)$ and $\bar{f}(i)$ - corresponding to the total weights of out-edges and in-edges of vertex $i$, respectively. Clearly, $f(\emptyset)=\bar{f}(\emptyset)=0$. Since $G$ is regular, $\rho(i|\emptyset) = f(i) = f(j) = \rho(j | \emptyset)$, $\bar{\rho}(i|\emptyset) = \bar{f}(i) = \bar{f}(j) = \bar{\rho}(j | \emptyset)$ for all $i,j \in V$. Consequently, every input item in this instance has identical ordering priority; and due to the algorithm being deterministic, the adversary can choose any permutation as a feasible input ordering. 

Denote by $k$ the number of initial steps taken by the algorithm that the adversary will anticipate. In other words, for $2^k$ possible partial solutions, the adversary prepares an input ordering (consistent with the algorithm's queries in these $k$ steps) such that any extendible solution has a bad approximation ratio.

\begin{theorem} For the unweighted Max-Di-Cut problem, no fixed order Q-Type 3 double-sided myopic algorithm can achieve an approximation ratio greater than $\frac{2}{3}$.
\end{theorem}

\begin{proof} Let $v_1,...,v_6$ be the vertices along a directed 6-cycle $G$ with unit edge weights. Clearly, $OPT=3$ is achieved by $\{v_1,v_3,v_5\}$ or $\{v_2,v_4,v_6\}$. The regularity of $G$ ensures that $\pi(v_1)=...=\pi(v_6)$ for any $\pi$, allowing the adversary to specify any input ordering. Suppose the algorithm accepts (\emph{resp.} rejects) $v_1$ in the first step, the adversary fixes $u_{i_2} = \{v_3,v_4\}$ as the next set of possible inputs in the sequence. At this point, $v_3$ and $v_4$ are indistinguishable to the algorithm, since $\rho(v_3 | S) = \rho(v_4 | S) = 1$ and $\bar\rho(v_3|S) = \bar\rho(v_4|S) = 1$ for any $S \subseteq \mathcal{N}_1 = \{v_1\}$. If the algorithm accepts $u_{i_2}$, then the adversary sets $u_{i_2} = v_4$ (\emph{resp.} $u_{i_2}=v_3$); otherwise it sets $u_{i_2} = v_3$ (\emph{resp.} $u_{i_2} = v_4$). Both cases contradict the optimal solutions, and the maximum cut value is now at most 2.

\end{proof}


Although this inapproximation bound does not match the currently best deterministic greedy $\frac{1}{3}$ approximation algorithm due to Buchbinder \emph{et al.} and Bar-Noy and Lampis for Max-Di-Cut, the inapproximation shows that the SDP based 0.828 approximation cannot be realized by fixed-order deterministic double sided myopic algorithms in our model.

\subsection{LP Construction for Q-Type 2}

We now generalize the adversarial strategy described in the previous section to obtain a sharper inapproximability ratio for the general unconstrained submodular maximization problem. First, we consider the fixed priority setting under query type 2.

\begin{lemma} \label{conditionslemma}
No fixed Q-Type 2 double-sided myopic algorithm can achieve an approximation ratio greater than $\frac{1}{c}$ if there exists a (non-negative) submodular function $f$ with the following conditions

\begin{enumerate}
	\item[{\bf Cond. 1}] $f(\{u\}) = f(\{v\}), \bar{f}(\{u\}) = \bar{f}(\{v\}), \forall u,v \in \mathcal{N}$. This imposes initial indistinguishably in the input set.
	\item[{\bf Cond. 2}] There exist subsets $A = \{a_1,...,a_k\} \subseteq \mathcal{N}$, $R = \{r_1,...,r_k\} \subseteq \mathcal{N}$, $A \cap R = \emptyset$, such that for every $1 \leq i < j \leq k$ and every $\mathcal{C}_i \in \{a_1,r_1\} \times \{a_2,r_2\} \times ... \times \{a_i,r_i\}$,
\begin{align*}
	f((\mathcal{C}_i\cap A) \cup \{a_{j}\}) =  f((\mathcal{C}_i\cap A) \cup \{r_{j}\}) \\
	\bar{f}((\mathcal{C}_i\cap R) \cup \{a_{j}\}) = \bar{f}((\mathcal{C}_i\cap R) \cup \{r_{j}\}) 
\end{align*}
	Although $\mathcal{C}_i$ is defined as an $i$-vector, we abuse notation and treat $\mathcal{C}_i$ as a set of size $i$. Semantically, $A$ (\emph{resp.} $R$) is the set of items that the algorithm is tricked into accepting (\emph{resp.} rejecting). This is achievable if $a_{j},r_{j}$ are indistinguishable to the algorithm at round $j$, in the Q-Type 2 restricted oracle model.

	\item[{\bf Cond. 3}] For every $\mathcal{C}_k \in \{a_1,r_1\} \times \{a_2,r_2\} \times ... \times \{a_k,r_k\}$, any solution $S \subseteq \mathcal{N}$ such that $\mathcal{C}_k \cap A \subseteq S$ and $S \cap \mathcal{C}_k \cap R  = \emptyset$ (\emph{i.e.} $S$ is an extension of $\mathcal{C}_k$) must have $f(S) \leq 1$. This ensures the algorithm does not recover after making $k$ decisions.
	\item[{\bf Cond. 4}] There exists a set $S^* \in 2^\mathcal{N}$ such that $f(S^*) \geq c$. 
\end{enumerate}

\end{lemma}

\begin{proof}
	By {\bf Cond. 1}, any permutation of the input items can be chosen by the adversary, since the algorithm assigns identical priority values to all items. The adversary starts by forcing the algorithm to only accept items in $A$, or reject items in $R$ in the first $k$ rounds. 
That is, if the algorithm accepts (\emph{resp.} rejects) in the $j^{th}$ round (for $j\leq k$), then the adversary chooses an ordering of $\mathcal{N}$ such that $a_j$ (\emph{resp.} $r_j$) is the $j^{th}$ item and $r_j$ (\emph{resp.} $a_j$) is the $l^{th}$ item, for some arbitrary $k < l \leq n$. To clarify, the adversary will choose either $a_j$ or $r_j$
as the $j^{th}$ item and then postpone the
item not chosen until after the first $k$ items have been decided upon. By induction, assume this is achievable in the first $j < k$ steps. Then there is a set of choices $\mathcal{C}_j \in \{a_1,r_1\} \times ... \times \{a_j,r_j\}$ such that $\mathcal{C}_j \cap A = X_j$ and $\mathcal{C}_j \cap R = Y_j$, where $X_j$ (\emph{resp.} $Y_j$) is the set of accepted (\emph{resp.} rejected) items so far. Consider what the algorithm knows about $u_{j+1} \in \{a_{j+1},r_{j+1}\}$ at this point. Using Q-Type 2 queries, the algorithm may learn the marginal difference in $f$ (\emph{resp.} $\bar{f}$) for $u_{j+1}$ w.r.t. all $X_i$ (\emph{resp.} $Y_i$) for $i \leq j$. However, due to {\bf Cond. 2}, the following must hold
\begin{align*}
	\rho(a_{j+1} | X_0) = \rho(r_{j+1}| X_0)&, \bar{\rho}(a_{j+1}| Y_0) = \bar{\rho}(r_{j+1}| Y_0) \\
	\vdots \\
	\rho(a_{j+1} | X_j) = \rho(r_{j+1} | X_j)&, \bar{\rho}(a_{j+1} | Y_j) = \bar{\rho}(r_{j+1}| Y_j )
\end{align*}

Since this is the only information available to the algorithm, the item $a_{j+1}$ and $r_{j+1}$ cannot be distinguished by the algorithm. Therefore if the algorithm accepts, then the adversary chooses $a_{j+1}$ as the $j+1^{st}$ input item, and $r_{j+1}$ otherwise. The base case follows from the same argument, as $C_0 = X_0 = Y_0 = \emptyset$. 

After $k$ steps, the algorithm constructs a partial solution (along with partial rejection) corresponding to some $\mathcal{C}_k \in \{a_1,r_1\} \times ... \times \{a_k,r_k\}$. By {\bf Cond. 3}, any complete solution that can be extended now has a value bounded by 1. Thus, with the optimum solution having a value of at least $c$ by {\bf Cond. 4}, we have forced the algorithm to return a solution with at most $\frac{1}{c}OPT$. 
\end{proof}

The conditions in Lemma~\ref{conditionslemma} can be expressed as a system of linear inequalities, which we formulate below as a linear program that maximizes the value of $c$. For fixed $n$ and $k$, we work with a ground set of size $n$ in the form $\mathcal{N} = \{s_1,...,s_{\lfloor \frac{n}{2} \rfloor}, o_1,...,o_{\lceil \frac{n}{2} \rceil} \}$, and designate the optimum set as $O = \{o_1,...,o_{\lceil \frac{n}{2} \rceil}\}$. Set $A = \{s_1,...,s_k\}$ and $R = \{o_1,...,o_k\}$ so as to deter the algorithm from the optimum set as much as possible. 
For every possible subset $S \subseteq \mathcal{N}$, we associate with it an LP variable $x_S$. As this construction entails exponentially many variables, this is indeed only feasible in practice when restricted to a small ground set. Semantically, the solution to each $x_S$ corresponds to the value of $f$ on $S$ - and as such we abuse notation and refer to $f(S)$ directly as the LP variable for $S$. For interpretability, we may refer to the variable $\bar{f}(S_i)$ as alias for the variable $f(\mathcal{N} \setminus S_i)$ (following the definition of $\bar{f}$). Define the linear program as follows: \\

\begin{mycapequ}[!ht]
\caption{\uline{LP for Q-Type 2 Fixed Priority}}
{\bf Objective}
\begin{align}
\max \quad
& f(\{o_1,...,o_ {\lceil \frac{n}{2} \rceil}\}) \label{objective}\tag{obj}
\end{align}
{\bf Constraints}
\begin{align}
f(S \cup \{v\}) + f(S \cup \{u\}) - f(S \cup \{v,u\}) &\geq  f(S)
&& \forall S \subseteq \mathcal{N}, (v,u) \in \mathcal{N}\setminus S\label{c1}\tag{C1} \\ \notag
f(\{v\}) - f(\{u\}) &=0, \\
\bar{f}(\{v\}) -\bar{f}(\{u\}) &= 0
&& \forall v,u \in \mathcal{N} \label{c2}\tag{C2} \\ \notag
f((\mathcal{C}_i \cap A)\cup \{a_{j}\}) -f((\mathcal{C}_i\cap A)\cup \{r_{j}\}) &= 0,\\
\bar{f}((\mathcal{C}_i \cap R) \cup \{a_{j}\}) - \bar{f}((\mathcal{C}_i \cap R) \cup \{r_{j}\}) &= 0 
&& \forall \mathcal{C}_i {\in} \{a_1,r_1\}{\times}...{\times}\{a_i,r_i\}, \notag\\ 
&&&0 < i<j \leq k \label{c3}\tag{C3} \\
f(S) &\leq 1 
&& \forall S \text{ s.t. } \exists \mathcal{C}_k, \mathcal{C}_k\cap A \subseteq S \notag \\
&&&\wedge S \cap \mathcal{C}_k \cap R = \emptyset \label{c4}\tag{C4} \\
f(S) &\geq 0 
&& \forall S \label{c5}\tag{C5} \\
f(\emptyset) & = 0 \label{c6}\tag{C6}
\end{align}
\end{mycapequ}

Inequality (\ref{c1}) is a necessary and sufficient condition for submodularity \cite{nemhauser}, and (\ref{c5}) and (\ref{c6}) constrain $f$ to be non-negative and normalized. The remaining constraints correspond to conditions 1-3 of Lemma~\ref{conditionslemma}. If the LP is feasible, then the objective value in (\ref{objective}) is a lower bound for $c$ in condition 4.

\begin{proof}[Proof of Theorem~\ref{fixedtheorem}]
	Running the LP with $n=8$ and $k=4$, a feasible solution is found with objective value of 2.3333. By Lemma~\ref{conditionslemma}, this demonstrates a lower bound of $\frac{1}{2.3333} \approx 0.428$. 
\end{proof}

\subsection{Extending to Q-Type 3}
Notice that in the fixed case, the difference between oracle query types is only reflected in the decision step, since the algorithm does not gain additional power in the ordering step. In other words, we can apply the same adversarial strategy to a stronger Q-Type as long as $a_i$ and $r_i$ are still indistinguishable. The following lemma extends the inapproximation result to fixed priority algorithms in the Q-Type 3 model by tightening the indistinguishability constraints. 

\begin{lemma} No Q-Type 3 fixed double-sided myopic algorithm can achieve an approximation ratio greater than $\frac{1}{c}$ if there exists a (non-negative) submodular function $f$ that satisfies {\bf Cond. 1-4} of Lemma~\ref{conditionslemma}, as well as the following:

\begin{enumerate}

\item[{\bf Cond. 2$\dagger$}] There exist subsets $A = \{a_1,...,a_k\} \subseteq \mathcal{N}$, $R = \{r_1,...,r_k\} \subseteq \mathcal{N}$, $A \cap R = \emptyset$, such that for every $0 < i < k$ and every $\mathcal{C}_i \in \{a_1,r_1\} \times \{a_2,r_2\} \times ... \times \{a_i,r_i\}$ and every subset $S \subseteq \mathcal{C}_i$,

\begin{align*}
	f(S \cup \{a_{i+1}\}) =  f(S \cup \{r_{i+1}\}) \\
	\bar{f}(S\cup \{a_{i+1}\}) = \bar{f}(S \cup \{r_{i+1}\}) 
\end{align*}

To decipher this statement, recall that each $\mathcal{C}_i$ corresponds to one of the valid input configurations in the first $i$ rounds. Then $\mathcal{N}_i = \mathcal{C}_i$ and the equality constraints simply forces $a_{i+1}$ and $r_{i+1}$ to have the same marginal differences that the algorithm is allowed to query under Q-Type 3. Notice also that this condition subsumes \textbf{Cond. 2} of the previous section, when we take $S = \mathcal{C}_i \cap A$ and $S = \mathcal{C}_i \cap R$. 

\end{enumerate}
\end{lemma}

\begin{proof}
Regardless of the relevant query type, a fixed algorithm must decide on a priority function based only on the singleton values. Therefore, as in the previous case, \textbf{Cond. 1} allows the adversary to introduce any input ordering. To adapt the proof for Lemma~\ref{conditionslemma} to Q-Type 3, it suffices to show that the new conditions are strong enough to push the induction step. Specifically, we would like $a_{i+1}$ and $r_{i+1}$ to have the same marginal descriptions at iteration $i+1$ assuming the adversary has been successful so far. This is easy to show, since any valid input corresponds to some $\mathcal{C}_i  \in \{a_1,r_1\} \times \{a_2,r_2\} \times ... \times \{a_i,r_i\}$, and this is precisely the set of items that the algorithm has seen so far. Under Q-Type 3, the algorithm may query the marginal value of the $i+1^{st}$ item with respect to any subset of $\mathcal{C}_i$ --- and so from \textbf{Cond. 2$\dagger$}, $a_{i+1}$ and $r_{i+1}$ are indistinguishable. The rest of the proof then follows identically.
\end{proof}

\begin{mycapequ}[!ht]
\caption{\uline{LP for Q-Type 3 Fixed Priority}}
\setcounter{equation}{0}
{\bf Objective}
\begin{align*}
\max \quad
& f(\{o_1,...,o_ {\lceil \frac{n}{2} \rceil}\})
\end{align*}
{\bf Constraints}
\begin{align*}
f(S \cup \{v\}) + f(S \cup \{u\}) - f(S \cup \{v,u\}) - f(S) &\geq 0
&& \forall S \subseteq \mathcal{N}, (v,u) \in \mathcal{N}\setminus S \\ \notag
f(S \cup \{a_{i+1}\}) - f(S \cup \{r_{i+1}\})  &= 0, \notag\\ \notag
\bar{f}(S\cup \{a_{i+1}\}) - \bar{f}(S \cup \{r_{i+1}\})  &= 0 
&& \forall S \subseteq \mathcal{C}_i, \\ 
&& &\forall \mathcal{C}_i \in \{a_1,r_1\}\times...\times\{a_i,r_i\}, \\ 
&& &0\leq i < k\\
f(S) &\leq 1 
&& \forall S \text{ s.t. } \exists \mathcal{C}_k, \mathcal{C}_k\cap A \subseteq S \notag \\
&&&\wedge S \cap \mathcal{C}_k \cap R = \emptyset \notag \\
f(S) &\geq 0 
&& \forall S \\
f(\emptyset) & = 0 \notag
\end{align*}
\end{mycapequ}

\begin{proof}[Proof of Theorem~\ref{fixedtheoremtype3}]
Running the LP with $n=8$ and $k=4$, a feasible solution is found with objective value of 2.2222. By Lemma~\ref{conditionslemma}, this demonstrates a lower bound of $\frac{1}{2.2222} \approx 0.450$. 
\end{proof}

\section{A 0.432 Inapproximation for Adaptive Priority Algorithms}
\label{adaptive}

In this section we extend the inapproximation to adaptive priority algorithms by applying a more restricted adversarial construction of that from the previous section. Surprisingly, this resulted in only a slight loss of tightness in the lower bound. 

\begin{lemma} No Q-Type 2 adaptive double-sided myopic algorithm can achieve an approximation ratio greater than $\frac{1}{c}$ if there exists a (non-negative) submodular function $f$ that satisfies {\bf Cond. 1-4} of Lemma~\ref{conditionslemma}, as well as the following:

\begin{enumerate}

\item[{\bf Cond. 2*}] There exist subsets $A = \{a_1,...,a_k\} \subseteq \mathcal{N}$, $R = \{r_1,...,r_k\} \subseteq \mathcal{N}$, $A \cap R = \emptyset$, such that for every $i < k$ and every $\mathcal{C}_i \in \{a_1,r_1\} \times \{a_2,r_2\} \times ... \times \{a_i,r_i\}$ and all pairs $v,u \in \mathcal{N} \setminus \mathcal{C}_i$,

\begin{align*}
	f((\mathcal{C}_i\cap A) \cup \{u\}) =  f((\mathcal{C}_i\cap A) \cup \{v\}) \\
	\bar{f}((\mathcal{C}_i\cap R) \cup \{u\}) = \bar{f}((\mathcal{C}_i\cap R) \cup \{v\}) 
\end{align*}

This condition subsumes both \textbf{Cond. 1} and {\bf Cond. 2} of Lemma~\ref{conditionslemma}, by requiring all unfixed items to have identical marginal descriptions - thus nullifying the benefit of adaptive ordering since the entire input will always be indistinguishable in the first $k$ rounds.

\end{enumerate}
\end{lemma}

\begin{proof}
Here, the algorithm is permitted to reorder the input at each iteration. The adversary responds by imposing indistinguishably on \emph{all} unfixed items in the first $k$ rounds. Consider the ordering step at the start of iteration $i<k$. By {\bf Cond. 2*}, then $\rho(u | X_j) = \rho(v | X_j), \bar{\rho}(u | Y_j) = \bar{\rho}(v | Y_j), j = 0,...,i-1$ for all $ u,v \in \mathcal{N}\setminus (X_{i-1} \cup Y_{i-1})$. This captures the full description of all unfixed items obtainable through Q-Type 2 value oracle access. Furthermore, any other information in the internal memory is independent of the remaining items, and thus provides no additional power. Hence, all items are indistinguishable and must be assigned the same priority. The adversary can now choose to place either $a_i$ or $r_i$ as the $i^{th}$ item, depending on the algorithm's decision. The rest of the proof is now identical to that of Lemma~\ref{conditionslemma}.
\end{proof}

\begin{mycapequ}[!ht]
\caption{\uline{LP for Q-Type 2 Adaptive Priority}}
\setcounter{equation}{0}
{\bf Objective}
\begin{align*}
\max \quad
& f(\{o_1,...,o_ {\lceil \frac{n}{2} \rceil}\}) \label{objective}
\end{align*}
{\bf Constraints}
\begin{align*}
f(S \cup \{v\}) + f(S \cup \{u\}) - f(S \cup \{v,u\}) - f(S) &\geq 0
&& \forall S \subseteq \mathcal{N}, (v,u) \in \mathcal{N}\setminus S \\ \notag
f((\mathcal{C}_i \cap A)\cup \{u\}) -f((\mathcal{C}_i\cap A)\cup \{v\}) &= 0, \notag\\ \notag
\bar{f}((\mathcal{C}_i \cap R) \cup \{u\}) - \bar{f}((\mathcal{C}_i \cap R) \cup \{v\}) &= 0 
&& \forall \mathcal{C}_i \in \{a_1,r_1\}\times...\times\{a_i,r_i\}, \\ 
&& &0\leq i < k, \;(u,v) \in \mathcal{N}\setminus \mathcal{C}_i \\
f(S) &\leq 1 
&& \forall S \text{ s.t. } \exists \mathcal{C}_k, \mathcal{C}_k\cap A \subseteq S \notag \\
&&&\wedge S \cap \mathcal{C}_k \cap R = \emptyset \notag \\
f(S) &\geq 0 
&& \forall S \\
f(\emptyset) & = 0 \notag
\end{align*}
\end{mycapequ}

\begin{proof}[Proof of Theorem~\ref{adaptivetheorem}]
Running the modified LP again using $n=8$ and $k=4$ produces a feasible solution with objective value of $c=2.3158$, giving us an inapproximability of $\frac{1}{c} \approx 0.432$.

\end{proof}

\section{Discussion of the Double-Sided Myopic Model and Open Problems}
\label{sec:discussion}

Adapting the priority framework \cite{Borodin:2002:PA:545381.545481}, we define the class of double-sided myopic algorithms and show how the double greedy algorithm of Buchbinder \emph{et al.} can be realized as an online double-sided myopic algorithm. We show that the double-sided interpretation of the double greedy algorithm of \cite{Buchbinder} satisfies the deterministic online model of Bar-Noy and Lampis \cite{barnoy:lampis}, for which they prove an online inapproximability of $\frac{2}{3 \sqrt{3}}$ for the Max-Di-Cut problem. As in Poloczek's \cite{poloczek2011bounds} priority inapproximation for Max-Sat, this provides evidence that the randomized $\frac{1}{2}$-approximation double greedy for USM cannot be de-randomized.

Our inapproximation follows from an LP formulation of possible algorithmic decisions, and at present does not yield a succinctly defined problem. However, we provide a $\frac{2}{3}$-inapproximation for Max-Di-Cut for fixed priority double-sided myopic algorithms. We wish to emphasize the generality of the myopic framework as it allows very general orderings to be defined by the algorithm, and does not impose any greedy aspect to the decisions (as to reject or accept an input) in each iteration of the algorithm. We also observe that \emph{non-greediness} appears to be essential in both the randomized double-greedy algorithm of Buchbinder \emph{et al.} as well as the deterministic approximation of Bar-Noy and Lampis for Max-Di-Cut on DAGs; while the deterministic double-greedy does make greedy decisions. 

The gap between the $\frac{1}{3}$ deterministic double greedy and the .432 inapproximation of our myopic framework remains open as does the gap between the $\frac{1}{3}$ approximation and known 
myopic inapproximations for explicit functions such as Max-Di-Cut. 
It is not clear how much further we can extend the LP formalization to improve our bounds, or if we can derive such bounds for succinct functions. We do not know if there is any provable difference between the online, fixed priority, and adaptive priority myopic models. In particular, can we improve upon the $\frac{1}{3}$ approximation for Max-Di-Cut by using a fixed or adaptive myopic algorithm? 

In addition to the above questions for deterministic algorithms, the next obvious direction is to establish limitations for 
\emph{randomized} double-sided myopic algorithms. Randomization can be utlilized in the ordering step and/or the decision step of the algorithm. Of particular interest is randomization in the decision step as in the 
Buchbinder et al algorithm. While we know that the $\frac{1}{2}$ approximation
 is tight under the value oracle model and under standard complexity
assumptions, it is still of interest to show such an inapproximation 
for myopic algorithms without any complexity constraints. 

\section*{Acknowledgments}
The authors would like to thank Yuval Filmus for the idea of employing linear programming, and Matthias Poloczek and Charles Rackoff for their comments and suggestions.

\newpage

\bibliographystyle{siam}
\bibliography{myrefs}{}

\newpage

\begin{appendices}

\begin{subappendices}
\section{LP Solution for Q-Type 3 Fixed Priority Inapproximation}
\label{fixedappendix}
We include for completeness the numerical solution of our LP formulation from Section~\ref{fixed} in Table~\ref{tab:allvalues}. Define the ground set $\mathcal{N} = \{1,2,3,4,5,6,7,8\}$, with the optimal solution $OPT = \{5,6,7,8\}$ and $f(OPT)=2.2222$. We set $k=4$, with $A = \{1,2,3,4\}$ and $R = \{5,6,7,8\}$. Specifically, the adversary chooses $i_1 \in \{1,5\}$, $i_2 \in \{2,6\}$, $i_3 \in \{3,7\}$, and $i_4 \in \{4,8\}$ depending on the algorithm's decisions; that is, the adversary sets $i_1 = 1$ (\emph{resp.} $i_1 = 5$) if the algorithm accepts (\emph{resp.} rejects) the first item, and so on.
To verify that our LP generated function $f$ supports the adversarial construction described in Section~\ref{adaptive}, it suffices to verify \textbf{Cond. 1, 2$\dagger$, 3} and \textbf{4}. Figure~\ref{fig:allmarginalsfixed} demonstrates that \textbf{Cond. 2$\dagger$} is satisfied: for each $0 \leq i < 4$ we show equality in the Q-Type 3 marginals\footnote{We show the value $f(\{u\} \cup S) = \rho(u|S) + f(S)$ instead of $\rho(u|S)$ as it allows for direct comparison with Table~\ref{tab:allvalues}. Clearly, $\rho(v|S) = \rho(u|S) \Leftrightarrow f(\{u\}\cup S) = f(\{v\} \cup S)$.}  between $a_{i+1}$ and $r_{i+1}$ for all allowable partial solutions. \textbf{Cond. 3} can be verified by examining Table~\ref{tab:allvalues} and observing that if $f(S) > 1$, then the subset $S$ is disallowed by the adversary. Specifically, for such a subset $S$, there is some $i \in [1,4]$ such that $a_i \not\in S$ and $r_i \in S$; this is a contradiction since the adversary always forces the algorithm to either accept $a_i$ or reject $r_i$. \textbf{Cond. 4} follows by setting $c=2.2222$, corresponding to the optimal solution.
Finally, we checked for submodularity separately by brute force, by verifying that $f(S\cup T) + f(S \cap T) \leq f(S) + f(T)$ for all $S,T \subseteq \mathcal{N}$. 

\pgfkeysifdefined{/pgfplots/table/output empty row/.@cmd}{
    \pgfplotstableset{
        empty header/.style={
          every head row/.style={output empty row},
        }
    }
}{
    \pgfplotstableset{
        empty header/.style={
            typeset cell/.append code={%
                \ifnum\pgfplotstablerow=-1 %
                    \pgfkeyssetvalue{/pgfplots/table/@cell content}{}%
                \fi
            }
        }
    }
}
{\footnotesize
\pgfplotstabletypeset[
    empty header,
    begin table=\begin{longtable},
    every first row/.append style={before row={%
    \caption{Complete description of $f:2^\mathcal{N} \to \mathbb{R}^+$ used in Theorem~\ref{fixedtheoremtype3}}%
    \label{tab:allvalues}\\\toprule
   \textbf{$S$} &\textbf{$f(S)$} & \textbf{$S$} &\textbf{$f(S)$} & \textbf{$S$} &\textbf{$f(S)$} & \textbf{$S$} &\textbf{$f(S)$} \\ \toprule    
    \endfirsthead
    \multicolumn{4}{c}%
    {{ Table \thetable\ Continued from previous page}} \\
    \toprule 
    \textbf{$S$} &\textbf{$f(S)$} & \textbf{$S$} &\textbf{$f(S)$} & \textbf{$S$} &\textbf{$f(S)$} & \textbf{$S$} &\textbf{$f(S)$} \\ \toprule  
    \endhead
    \midrule \multicolumn{2}{r}{{Continued on next page}} \\ \bottomrule
    \endfoot
    \midrule
    \multicolumn{2}{r}{{Concluded}} \\ \bottomrule
    \endlastfoot
    }},%
    end table=\end{longtable},
    col sep=ampersand,
    string type,
	columns={S,f,S,f,S,f,S,f},
    display columns/0/.style={
        select equal part entry of={0}{4},
        string type,
        column name={$S$},
        column type={r}, 
    },
    display columns/1/.style={
        select equal part entry of={0}{4},
        string type,
        column name={$f(S)$},
        column type={l|},
    },
    display columns/2/.style={select equal part entry of={1}{4},string type, column name=$S$,column type={r}},
    display columns/3/.style={select equal part entry of={1}{4},string type, column name=$f(S)$, column type={l|}},
   display columns/4/.style={select equal part entry of={2}{4},string type, column name=$S$,column type={r}},
    display columns/5/.style={select equal part entry of={2}{4},string type, column name=$f(S)$,column type={l|}},
	  display columns/6/.style={select equal part entry of={3}{4},string type, column name=$S$,column type={r}},
    display columns/7/.style={select equal part entry of={3}{4},string type, column name=$f(S)$,column type={l}},
    ]{datafixed.csv}
}

\begin{figure}%
\caption{Certificate of \textbf{Cond. 2$\dagger$} for Q-Type 3 fixed priority inapproximation}
\label{fig:allmarginalsfixed}
\begin{subfigure}[b]{\textwidth}
  \centering
  \begin{subfigure}[b]{0.3\textwidth} 
{\tiny
\begin{tabular}{|l|} \hline \multicolumn{1}{|c|}{$X_1 : \{  \},   Y_1 : \{ 5 \}$}\\ \hline
$f( 5 2 )= 1.111111$ \\
$f( 5 6 )= 1.111111$ \\
$\bar f( 5 2 )= 0.8888889$ \\
$\bar f( 5 6 )= 0.8888889$ \\

\hline  \end{tabular}}
\end{subfigure}
\begin{subfigure}[b]{0.3\textwidth}
{\tiny
\begin{tabular}{|l|} \hline \multicolumn{1}{|c|}{$X_1 : \{ 1 \},   Y_1 : \{  \}$}\\ \hline
$f( 1 2 )= 0.8888889$ \\
$f( 1 6 )= 0.8888889$ \\
$\bar f( 1 2 )= 1.111111$ \\
$\bar f( 1 6 )= 1.111111$ \\
\hline  \end{tabular}}
\end{subfigure}
\caption{All allowable marginal descriptions for $a_2$ and $r_2$}
\end{subfigure}
\begin{subfigure}[b]{\textwidth}
  \centering
  \begin{subfigure}[b]{0.2\textwidth} 
{\tiny
\begin{tabular}{|l|} \hline \multicolumn{1}{|c|}{$X_2 : \{  \},   Y_2 : \{ 5 6 \}$}\\ \hline
$f( 6 3 )= 1.111111$ \\
$f( 6 7 )= 1.111111$ \\
$\bar f( 6 3 )= 0.8888889$ \\
$\bar f( 6 7 )= 0.8888889$ \\
$f( 5 3 )= 1.111111$ \\
$f( 5 7 )= 1.111111$ \\
$\bar f( 5 3 )= 0.8888889$ \\
$\bar f( 5 7 )= 0.8888889$ \\
$f( 5 6 3 )= 1.666667$ \\
$f( 5 6 7 )= 1.666667$ \\
$\bar f( 5 6 3 )= 1$ \\
$\bar f( 5 6 7 )= 1$ \\
\hline  \end{tabular}}
\end{subfigure}
\begin{subfigure}[b]{0.2\textwidth}
{\tiny
\begin{tabular}{|l|} \hline \multicolumn{1}{|c|}{$X_2 : \{ 2 \},   Y_2 : \{ 5 \}$}\\ \hline
$f( 5 3 )= 1.111111$ \\
$f( 5 7 )= 1.111111$ \\
$\bar f( 5 3 )= 0.8888889$ \\
$\bar f( 5 7 )= 0.8888889$ \\
$f( 2 3 )= 0.8888889$ \\
$f( 2 7 )= 0.8888889$ \\
$\bar f( 2 3 )= 1.111111$ \\
$\bar f( 2 7 )= 1.111111$ \\
$f( 2 5 3 )= 1.444444$ \\
$f( 2 5 7 )= 1.444444$ \\
$\bar f( 2 5 3 )= 1.222222$ \\
$\bar f( 2 5 7 )= 1.222222$ \\
\hline  \end{tabular}}
\end{subfigure}
\begin{subfigure}[b]{0.2\textwidth}
{\tiny
\begin{tabular}{|l|} \hline \multicolumn{1}{|c|}{$X_2 : \{ 1 \},   Y_2 : \{ 6 \}$}\\ \hline
$f( 6 3 )= 1.111111$ \\
$f( 6 7 )= 1.111111$ \\
$\bar f( 6 3 )= 0.8888889$ \\
$\bar f( 6 7 )= 0.8888889$ \\
$f( 1 3 )= 0.8888889$ \\
$f( 1 7 )= 0.8888889$ \\
$\bar f( 1 3 )= 1.111111$ \\
$\bar f( 1 7 )= 1.111111$ \\
$f( 1 6 3 )= 1.222222$ \\
$f( 1 6 7 )= 1.222222$ \\
$\bar f( 1 6 3 )= 1.333333$ \\
$\bar f( 1 6 7 )= 1.333333$ \\
\hline  \end{tabular}}
\end{subfigure}
\begin{subfigure}[b]{0.2\textwidth}
{\tiny
\begin{tabular}{|l|} \hline \multicolumn{1}{|c|}{$X_2 : \{ 1 2 \},   Y_2 : \{  \}$}\\ \hline
$f( 2 3 )= 0.8888889$ \\
$f( 2 7 )= 0.8888889$ \\
$\bar f( 2 3 )= 1.111111$ \\
$\bar f( 2 7 )= 1.111111$ \\
$f( 1 3 )= 0.8888889$ \\
$f( 1 7 )= 0.8888889$ \\
$\bar f( 1 3 )= 1.111111$ \\
$\bar f( 1 7 )= 1.111111$ \\
$f( 1 2 3 )= 1$ \\
$f( 1 2 7 )= 1$ \\
$\bar f( 1 2 3 )= 1.666667$ \\
$\bar f( 1 2 7 )= 1.666667$ \\
\hline  \end{tabular}}
\end{subfigure}
\caption{All allowable marginal descriptions for  $a_3$ and $r_3$}

\begin{subfigure}[b]{\textwidth}
  \centering
\begin{subfigure}[b]{0.2\textwidth}{\tiny  \begin{tabular}{|l|} \hline \multicolumn{1}{|c|}{$X_3 : \{  \},   Y_3 : \{ 5 6 7 \}$}\\ \hline
$f( 7 4 )= 1.111111$ \\
$f( 7 8 )= 1.111111$ \\
$\bar f( 7 4 )= 0.7222222$ \\
$\bar f( 7 8 )= 0.7222222$ \\
$f( 6 4 )= 1.111111$ \\
$f( 6 8 )= 1.111111$ \\
$\bar f( 6 4 )= 0.7777778$ \\
$\bar f( 6 8 )= 0.7777778$ \\
$f( 6 7 4 )= 1.666667$ \\
$f( 6 7 8 )= 1.666667$ \\
$\bar f( 6 7 4 )= 0.9444444$ \\
$\bar f( 6 7 8 )= 0.9444444$ \\
$f( 5 4 )= 1.111111$ \\
$f( 5 8 )= 1.111111$ \\
$\bar f( 5 4 )= 0.7777778$ \\
$\bar f( 5 8 )= 0.7777778$ \\
$f( 5 7 4 )= 1.666667$ \\
$f( 5 7 8 )= 1.666667$ \\
$\bar f( 5 7 4 )= 0.8888889$ \\
$\bar f( 5 7 8 )= 0.8888889$ \\
$f( 5 6 4 )= 1.666667$ \\
$f( 5 6 8 )= 1.666667$ \\
$\bar f( 5 6 4 )= 1$ \\
$\bar f( 5 6 8 )= 1$ \\
$f( 5 6 7 4 )= 2.222222$ \\
$f( 5 6 7 8 )= 2.222222$ \\
$\bar f( 5 6 7 4 )= 1$ \\
$\bar f( 5 6 7 8 )= 1$ \\

\hline  \end{tabular}} \end{subfigure} \begin{subfigure}[b]{0.2\textwidth}{\tiny  \begin{tabular}{|l|} \hline \multicolumn{1}{|c|}{$X_3 : \{ 3 \},   Y_3 : \{ 5 6 \}$}\\ \hline
$f( 6 4 )= 1.111111$ \\
$f( 6 8 )= 1.111111$ \\
$\bar f( 6 4 )= 0.7777778$ \\
$\bar f( 6 8 )= 0.7777778$ \\
$f( 5 4 )= 1.111111$ \\
$f( 5 8 )= 1.111111$ \\
$\bar f( 5 4 )= 0.7777778$ \\
$\bar f( 5 8 )= 0.7777778$ \\
$f( 5 6 4 )= 1.666667$ \\
$f( 5 6 8 )= 1.666667$ \\
$\bar f( 5 6 4 )= 1$ \\
$\bar f( 5 6 8 )= 1$ \\
$f( 3 4 )= 0.8888889$ \\
$f( 3 8 )= 0.8888889$ \\
$\bar f( 3 4 )= 1.111111$ \\
$\bar f( 3 8 )= 1.111111$ \\
$f( 3 6 4 )= 1.222222$ \\
$f( 3 6 8 )= 1.222222$ \\
$\bar f( 3 6 4 )= 1.111111$ \\
$\bar f( 3 6 8 )= 1.111111$ \\
$f( 3 5 4 )= 1.444444$ \\
$f( 3 5 8 )= 1.444444$ \\
$\bar f( 3 5 4 )= 1.111111$ \\
$\bar f( 3 5 8 )= 1.111111$ \\
$f( 3 5 6 4 )= 1.777778$ \\
$f( 3 5 6 8 )= 1.777778$ \\
$\bar f( 3 5 6 4 )= 1.111111$ \\
$\bar f( 3 5 6 8 )= 1.111111$ \\

\hline  \end{tabular}} \end{subfigure} \begin{subfigure}[b]{0.2\textwidth}{\tiny  \begin{tabular}{|l|} \hline \multicolumn{1}{|c|}{$X_3 : \{ 2 \},   Y_3 : \{ 5 7 \}$}\\ \hline
$f( 7 4 )= 1.111111$ \\
$f( 7 8 )= 1.111111$ \\
$\bar f( 7 4 )= 0.7222222$ \\
$\bar f( 7 8 )= 0.7222222$ \\
$f( 5 4 )= 1.111111$ \\
$f( 5 8 )= 1.111111$ \\
$\bar f( 5 4 )= 0.7777778$ \\
$\bar f( 5 8 )= 0.7777778$ \\
$f( 5 7 4 )= 1.666667$ \\
$f( 5 7 8 )= 1.666667$ \\
$\bar f( 5 7 4 )= 0.8888889$ \\
$\bar f( 5 7 8 )= 0.8888889$ \\
$f( 2 4 )= 0.7777778$ \\
$f( 2 8 )= 0.7777778$ \\
$\bar f( 2 4 )= 1.111111$ \\
$\bar f( 2 8 )= 1.111111$ \\
$f( 2 7 4 )= 1.111111$ \\
$f( 2 7 8 )= 1.111111$ \\
$\bar f( 2 7 4 )= 1.222222$ \\
$\bar f( 2 7 8 )= 1.222222$ \\
$f( 2 5 4 )= 1.333333$ \\
$f( 2 5 8 )= 1.333333$ \\
$\bar f( 2 5 4 )= 1.111111$ \\
$\bar f( 2 5 8 )= 1.111111$ \\
$f( 2 5 7 4 )= 1.555556$ \\
$f( 2 5 7 8 )= 1.555556$ \\
$\bar f( 2 5 7 4 )= 1.222222$ \\
$\bar f( 2 5 7 8 )= 1.222222$ \\

\hline  \end{tabular}} \end{subfigure} \begin{subfigure}[b]{0.2\textwidth}{\tiny  \begin{tabular}{|l|} \hline \multicolumn{1}{|c|}{$X_3 : \{ 2 3 \},   Y_3 : \{ 5 \}$}\\ \hline
$f( 5 4 )= 1.111111$ \\
$f( 5 8 )= 1.111111$ \\
$\bar f( 5 4 )= 0.7777778$ \\
$\bar f( 5 8 )= 0.7777778$ \\
$f( 3 4 )= 0.8888889$ \\
$f( 3 8 )= 0.8888889$ \\
$\bar f( 3 4 )= 1.111111$ \\
$\bar f( 3 8 )= 1.111111$ \\
$f( 3 5 4 )= 1.444444$ \\
$f( 3 5 8 )= 1.444444$ \\
$\bar f( 3 5 4 )= 1.111111$ \\
$\bar f( 3 5 8 )= 1.111111$ \\
$f( 2 4 )= 0.7777778$ \\
$f( 2 8 )= 0.7777778$ \\
$\bar f( 2 4 )= 1.111111$ \\
$\bar f( 2 8 )= 1.111111$ \\
$f( 2 5 4 )= 1.333333$ \\
$f( 2 5 8 )= 1.333333$ \\
$\bar f( 2 5 4 )= 1.111111$ \\
$\bar f( 2 5 8 )= 1.111111$ \\
$f( 2 3 4 )= 0.8888889$ \\
$f( 2 3 8 )= 0.8888889$ \\
$\bar f( 2 3 4 )= 1.666667$ \\
$\bar f( 2 3 8 )= 1.666667$ \\
$f( 2 3 5 4 )= 1.388889$ \\
$f( 2 3 5 8 )= 1.388889$ \\
$\bar f( 2 3 5 4 )= 1.444444$ \\
$\bar f( 2 3 5 8 )= 1.444444$ \\

\hline  \end{tabular}} \end{subfigure} \begin{subfigure}[b]{0.2\textwidth}{\tiny  \begin{tabular}{|l|} \hline \multicolumn{1}{|c|}{$X_3 : \{ 1 \},   Y_3 : \{ 6 7 \}$}\\ \hline
$f( 7 4 )= 1.111111$ \\
$f( 7 8 )= 1.111111$ \\
$\bar f( 7 4 )= 0.7222222$ \\
$\bar f( 7 8 )= 0.7222222$ \\
$f( 6 4 )= 1.111111$ \\
$f( 6 8 )= 1.111111$ \\
$\bar f( 6 4 )= 0.7777778$ \\
$\bar f( 6 8 )= 0.7777778$ \\
$f( 6 7 4 )= 1.666667$ \\
$f( 6 7 8 )= 1.666667$ \\
$\bar f( 6 7 4 )= 0.9444444$ \\
$\bar f( 6 7 8 )= 0.9444444$ \\
$f( 1 4 )= 0.7777778$ \\
$f( 1 8 )= 0.7777778$ \\
$\bar f( 1 4 )= 1.111111$ \\
$\bar f( 1 8 )= 1.111111$ \\
$f( 1 7 4 )= 1.111111$ \\
$f( 1 7 8 )= 1.111111$ \\
$\bar f( 1 7 4 )= 1.277778$ \\
$\bar f( 1 7 8 )= 1.277778$ \\
$f( 1 6 4 )= 1.111111$ \\
$f( 1 6 8 )= 1.111111$ \\
$\bar f( 1 6 4 )= 1.333333$ \\
$\bar f( 1 6 8 )= 1.333333$ \\
$f( 1 6 7 4 )= 1.444444$ \\
$f( 1 6 7 8 )= 1.444444$ \\
$\bar f( 1 6 7 4 )= 1.388889$ \\
$\bar f( 1 6 7 8 )= 1.388889$ \\

\hline  \end{tabular}} \end{subfigure} \begin{subfigure}[b]{0.2\textwidth}{\tiny  \begin{tabular}{|l|} \hline \multicolumn{1}{|c|}{$X_3 : \{ 1 3 \},   Y_3 : \{ 6 \}$}\\ \hline
$f( 6 4 )= 1.111111$ \\
$f( 6 8 )= 1.111111$ \\
$\bar f( 6 4 )= 0.7777778$ \\
$\bar f( 6 8 )= 0.7777778$ \\
$f( 3 4 )= 0.8888889$ \\
$f( 3 8 )= 0.8888889$ \\
$\bar f( 3 4 )= 1.111111$ \\
$\bar f( 3 8 )= 1.111111$ \\
$f( 3 6 4 )= 1.222222$ \\
$f( 3 6 8 )= 1.222222$ \\
$\bar f( 3 6 4 )= 1.111111$ \\
$\bar f( 3 6 8 )= 1.111111$ \\
$f( 1 4 )= 0.7777778$ \\
$f( 1 8 )= 0.7777778$ \\
$\bar f( 1 4 )= 1.111111$ \\
$\bar f( 1 8 )= 1.111111$ \\
$f( 1 6 4 )= 1.111111$ \\
$f( 1 6 8 )= 1.111111$ \\
$\bar f( 1 6 4 )= 1.333333$ \\
$\bar f( 1 6 8 )= 1.333333$ \\
$f( 1 3 4 )= 1$ \\
$f( 1 3 8 )= 1$ \\
$\bar f( 1 3 4 )= 1.666667$ \\
$\bar f( 1 3 8 )= 1.666667$ \\
$f( 1 3 6 4 )= 1.222222$ \\
$f( 1 3 6 8 )= 1.222222$ \\
$\bar f( 1 3 6 4 )= 1.555556$ \\
$\bar f( 1 3 6 8 )= 1.555556$ \\

\hline  \end{tabular}} \end{subfigure} \begin{subfigure}[b]{0.2\textwidth}{\tiny  \begin{tabular}{|l|} \hline \multicolumn{1}{|c|}{$X_3 : \{ 1 2 \},   Y_3 : \{ 7 \}$}\\ \hline
$f( 7 4 )= 1.111111$ \\
$f( 7 8 )= 1.111111$ \\
$\bar f( 7 4 )= 0.7222222$ \\
$\bar f( 7 8 )= 0.7222222$ \\
$f( 2 4 )= 0.7777778$ \\
$f( 2 8 )= 0.7777778$ \\
$\bar f( 2 4 )= 1.111111$ \\
$\bar f( 2 8 )= 1.111111$ \\
$f( 2 7 4 )= 1.111111$ \\
$f( 2 7 8 )= 1.111111$ \\
$\bar f( 2 7 4 )= 1.222222$ \\
$\bar f( 2 7 8 )= 1.222222$ \\
$f( 1 4 )= 0.7777778$ \\
$f( 1 8 )= 0.7777778$ \\
$\bar f( 1 4 )= 1.111111$ \\
$\bar f( 1 8 )= 1.111111$ \\
$f( 1 7 4 )= 1.111111$ \\
$f( 1 7 8 )= 1.111111$ \\
$\bar f( 1 7 4 )= 1.277778$ \\
$\bar f( 1 7 8 )= 1.277778$ \\
$f( 1 2 4 )= 1$ \\
$f( 1 2 8 )= 1$ \\
$\bar f( 1 2 4 )= 1.666667$ \\
$\bar f( 1 2 8 )= 1.666667$ \\
$f( 1 2 7 4 )= 1.111111$ \\
$f( 1 2 7 8 )= 1.111111$ \\
$\bar f( 1 2 7 4 )= 1.777778$ \\
$\bar f( 1 2 7 8 )= 1.777778$ \\

\hline  \end{tabular}} \end{subfigure} \begin{subfigure}[b]{0.2\textwidth}{\tiny  \begin{tabular}{|l|} \hline \multicolumn{1}{|c|}{$X_3 : \{ 1 2 3 \},   Y_3 : \{  \}$}\\ \hline
$f( 3 4 )= 0.8888889$ \\
$f( 3 8 )= 0.8888889$ \\
$\bar f( 3 4 )= 1.111111$ \\
$\bar f( 3 8 )= 1.111111$ \\
$f( 2 4 )= 0.7777778$ \\
$f( 2 8 )= 0.7777778$ \\
$\bar f( 2 4 )= 1.111111$ \\
$\bar f( 2 8 )= 1.111111$ \\
$f( 2 3 4 )= 0.8888889$ \\
$f( 2 3 8 )= 0.8888889$ \\
$\bar f( 2 3 4 )= 1.666667$ \\
$\bar f( 2 3 8 )= 1.666667$ \\
$f( 1 4 )= 0.7777778$ \\
$f( 1 8 )= 0.7777778$ \\
$\bar f( 1 4 )= 1.111111$ \\
$\bar f( 1 8 )= 1.111111$ \\
$f( 1 3 4 )= 1$ \\
$f( 1 3 8 )= 1$ \\
$\bar f( 1 3 4 )= 1.666667$ \\
$\bar f( 1 3 8 )= 1.666667$ \\
$f( 1 2 4 )= 1$ \\
$f( 1 2 8 )= 1$ \\
$\bar f( 1 2 4 )= 1.666667$ \\
$\bar f( 1 2 8 )= 1.666667$ \\
$f( 1 2 3 4 )= 1$ \\
$f( 1 2 3 8 )= 1$ \\
$\bar f( 1 2 3 4 )= 2.222222$ \\
$\bar f( 1 2 3 8 )= 2.222222$ \\
\hline  \end{tabular}} 
\end{subfigure}
\end{subfigure}
\caption{All allowable marginal descriptions for  $a_4$ and $r_4$}
\end{subfigure}
\end{figure}

\section{LP Solution for Q-Type 2 Adaptive Priority Inapproximation}
Similarly, we include the numerical solution of our LP formulation for the adaptive priority case. We employ identical notations as that of Appendix~\ref{fixedappendix}. In this case, indistinguishability is required between \emph{all} unfixed items with respect to Q-Type 2 (as in \textbf{Cond. 2*}). We show this in Table~\ref{tab:allmarginals}, in which we enumerate all configurations in the first 4 steps, and for each configuration show equality in the marginal differences of all remaining items. The verification of the remaining conditions follows the description in Appendix~\ref{fixedappendix}. In particular, observe that $f(OPT) = f(\{5,6,7,8\}) = 2.3158$ gives us the claimed bound in Theorme~\ref{adaptivetheorem}.

\pgfkeysifdefined{/pgfplots/table/output empty row/.@cmd}{
    \pgfplotstableset{
        empty header/.style={
          every head row/.style={output empty row},
        }
    }
}{
    \pgfplotstableset{
        empty header/.style={
            typeset cell/.append code={%
                \ifnum\pgfplotstablerow=-1 %
                    \pgfkeyssetvalue{/pgfplots/table/@cell content}{}%
                \fi
            }
        }
    }
}
{\footnotesize
\pgfplotstabletypeset[
    empty header,
    begin table=\begin{longtable},
    every first row/.append style={before row={%
    \caption{Complete description of $f:2^\mathcal{N} \to \mathbb{R}^+$ used in Theorem~\ref{adaptivetheorem}}%
    \label{tab:allvalues}\\\toprule
   \textbf{$S$} &\textbf{$f(S)$} & \textbf{$S$} &\textbf{$f(S)$} & \textbf{$S$} &\textbf{$f(S)$} & \textbf{$S$} &\textbf{$f(S)$} \\ \toprule    
    \endfirsthead
    \multicolumn{4}{c}%
    {{ Table \thetable\ Continued from previous page}} \\
    \toprule 
    \textbf{$S$} &\textbf{$f(S)$} & \textbf{$S$} &\textbf{$f(S)$} & \textbf{$S$} &\textbf{$f(S)$} & \textbf{$S$} &\textbf{$f(S)$} \\ \toprule  
    \endhead
    \midrule \multicolumn{2}{r}{{Continued on next page}} \\ \bottomrule
    \endfoot
    \midrule
    \multicolumn{2}{r}{{Concluded}} \\ \bottomrule
    \endlastfoot
    }},%
    end table=\end{longtable},
    col sep=ampersand,
    string type,
	columns={S,f,S,f,S,f,S,f},
    display columns/0/.style={
        select equal part entry of={0}{4},
        string type,
        column name={$S$},
        column type={r}, 
    },
    display columns/1/.style={
        select equal part entry of={0}{4},
        string type,
        column name={$f(S)$},
        column type={l|},
    },
    display columns/2/.style={select equal part entry of={1}{4},string type, column name=$S$,column type={r}},
    display columns/3/.style={select equal part entry of={1}{4},string type, column name=$f(S)$, column type={l|}},
   display columns/4/.style={select equal part entry of={2}{4},string type, column name=$S$,column type={r}},
    display columns/5/.style={select equal part entry of={2}{4},string type, column name=$f(S)$,column type={l|}},
	  display columns/6/.style={select equal part entry of={3}{4},string type, column name=$S$,column type={r}},
    display columns/7/.style={select equal part entry of={3}{4},string type, column name=$f(S)$,column type={l}},
    ]{data.csv}
}

\begin{table}[!htb]
    \caption{Indistinguishability between all unfixed items in the first 3 iterations with respect to Q-Type 2.} \label{tab:allmarginals}
	{\tiny

\begin{subtable}{.2\linewidth} \centering \begin{tabular}{|l|} \hline
\multicolumn{1}{|c|}{$X_0 = \{  \},   Y_0 = \{  \}$} \\ \hline
$f( 1 )= 0.5789474 $\\
$\bar f( 1 )= 0.5789474 $\\
$f( 2 )= 0.5789474 $\\
$\bar f( 2 )= 0.5789474 $\\
$f( 3 )= 0.5789474 $\\
$\bar f( 3 )= 0.5789474 $\\
$f( 4 )= 0.5789474 $\\
$\bar f( 4 )= 0.5789474 $\\
$f( 5 )= 0.5789474 $\\
$\bar f( 5 )= 0.5789474 $\\
$f( 6 )= 0.5789474 $\\
$\bar f( 6 )= 0.5789474 $\\
$f( 7 )= 0.5789474 $\\
$\bar f( 7 )= 0.5789474 $\\
$f( 8 )= 0.5789474 $\\
$\bar f( 8 )= 0.5789474 $\\
\hline  \end{tabular} \end{subtable}
\begin{subtable}{.2\linewidth} \centering \begin{tabular}{|l|} \hline
\multicolumn{1}{|c|}{$X_1 = \{ 1 \},   Y_1 = \{  \}$} \\ \hline
$f( 1 2 )= 0.8947368 $\\
$\bar f( 2 )= 0.5789474 $\\
$f( 1 3 )= 0.8947368 $\\
$\bar f( 3 )= 0.5789474 $\\
$f( 1 4 )= 0.8947368 $\\
$\bar f( 4 )= 0.5789474 $\\
$f( 1 5 )= 0.8947368 $\\
$\bar f( 5 )= 0.5789474 $\\
$f( 1 6 )= 0.8947368 $\\
$\bar f( 6 )= 0.5789474 $\\
$f( 1 7 )= 0.8947368 $\\
$\bar f( 7 )= 0.5789474 $\\
$f( 1 8 )= 0.8947368 $\\
$\bar f( 8 )= 0.5789474 $\\
\hline  \end{tabular} \end{subtable}
\begin{subtable}{.2\linewidth} \centering \begin{tabular}{|l|} \hline
\multicolumn{1}{|c|}{$X_1 = \{ 1,2 \},   Y_1 = \{  \}$} \\ \hline
$f( 1 2 3 )= 1 $\\
$\bar f( 3 )= 0.5789474 $\\
$f( 1 2 4 )= 1 $\\
$\bar f( 4 )= 0.5789474 $\\
$f( 1 2 5 )= 1 $\\
$\bar f( 5 )= 0.5789474 $\\
$f( 1 2 6 )= 1 $\\
$\bar f( 6 )= 0.5789474 $\\
$f( 1 2 7 )= 1 $\\
$\bar f( 7 )= 0.5789474 $\\
$f( 1 2 8 )= 1 $\\
$\bar f( 8 )= 0.5789474 $\\
\hline  \end{tabular} \end{subtable}
\begin{subtable}{.2\linewidth} \centering \begin{tabular}{|l|} \hline
\multicolumn{1}{|c|}{$X_2 = \{ 1,2,3 \},   Y_2 = \{  \}$} \\ \hline
$f( 1 2 3 4 )= 1 $\\
$\bar f( 4 )= 0.5789474 $\\
$f( 1 2 3 5 )= 1 $\\
$\bar f( 5 )= 0.5789474 $\\
$f( 1 2 3 6 )= 1 $\\
$\bar f( 6 )= 0.5789474 $\\
$f( 1 2 3 7 )= 1 $\\
$\bar f( 7 )= 0.5789474 $\\
$f( 1 2 3 8 )= 1 $\\
$\bar f( 8 )= 0.5789474 $\\
\hline  \end{tabular} \end{subtable}
 \begin{subtable}{.2\linewidth} \centering \begin{tabular}{|l|} \hline
\multicolumn{1}{|c|}{$X_2 = \{ 1,2 \},   Y_2 = \{ 7 \}$} \\ \hline
$f( 1 2 3 )= 1 $\\
$\bar f( 7 3 )= 0.8947368 $\\
$f( 1 2 4 )= 1 $\\
$\bar f( 7 4 )= 0.8947368 $\\
$f( 1 2 5 )= 1 $\\
$\bar f( 7 5 )= 0.8947368 $\\
$f( 1 2 6 )= 1 $\\
$\bar f( 7 6 )= 0.8947368 $\\
$f( 1 2 8 )= 1 $\\
$\bar f( 7 8 )= 0.8947368 $\\
\hline  \end{tabular} \end{subtable}
\begin{subtable}{.2\linewidth} \centering \begin{tabular}{|l|} \hline
\multicolumn{1}{|c|}{$X_2 = \{ 1 \},   Y_2 = \{ 6 \}$} \\ \hline
$f( 1 2 )= 0.8947368 $\\
$\bar f( 6 2 )= 0.8947368 $\\
$f( 1 3 )= 0.8947368 $\\
$\bar f( 6 3 )= 0.8947368 $\\
$f( 1 4 )= 0.8947368 $\\
$\bar f( 6 4 )= 0.8947368 $\\
$f( 1 5 )= 0.8947368 $\\
$\bar f( 6 5 )= 0.8947368 $\\
$f( 1 7 )= 0.8947368 $\\
$\bar f( 6 7 )= 0.8947368 $\\
$f( 1 8 )= 0.8947368 $\\
$\bar f( 6 8 )= 0.8947368 $\\
\hline  \end{tabular} \end{subtable}
\begin{subtable}{.2\linewidth} \centering \begin{tabular}{|l|} \hline
\multicolumn{1}{|c|}{$X_3 = \{ 1, 3 \},   Y_3 = \{ 6 \}$} \\ \hline
$f( 1 3 2 )= 1 $\\
$\bar f( 6 2 )= 0.8947368 $\\
$f( 1 3 4 )= 1 $\\
$\bar f( 6 4 )= 0.8947368 $\\
$f( 1 3 5 )= 1 $\\
$\bar f( 6 5 )= 0.8947368 $\\
$f( 1 3 7 )= 1 $\\
$\bar f( 6 7 )= 0.8947368 $\\
$f( 1 3 8 )= 1 $\\
$\bar f( 6 8 )= 0.8947368 $\\
\hline  \end{tabular} \end{subtable}
\begin{subtable}{.2\linewidth} \centering \begin{tabular}{|l|} \hline
\multicolumn{1}{|c|}{$X_3 = \{ 1 \},   Y_3 = \{ 6,7 \}$} \\ \hline
$f( 1 2 )= 0.8947368 $\\
$\bar f( 6 7 2 )= 1 $\\
$f( 1 3 )= 0.8947368 $\\
$\bar f( 6 7 3 )= 1 $\\
$f( 1 4 )= 0.8947368 $\\
$\bar f( 6 7 4 )= 1 $\\
$f( 1 5 )= 0.8947368 $\\
$\bar f( 6 7 5 )= 1 $\\
$f( 1 8 )= 0.8947368 $\\
$\bar f( 6 7 8 )= 1 $\\
\hline  \end{tabular} \end{subtable} 
\begin{subtable}{.2\linewidth} \centering \begin{tabular}{|l|} \hline
\multicolumn{1}{|c|}{$X_3 = \{ \},   Y_3 = \{ 5 \}$} \\ \hline
$f( 1 )= 0.5789474 $\\
$\bar f( 5 1 )= 0.8947368 $\\
$f( 2 )= 0.5789474 $\\
$\bar f( 5 2 )= 0.8947368 $\\
$f( 3 )= 0.5789474 $\\
$\bar f( 5 3 )= 0.8947368 $\\
$f( 4 )= 0.5789474 $\\
$\bar f( 5 4 )= 0.8947368 $\\
$f( 6 )= 0.5789474 $\\
$\bar f( 5 6 )= 0.8947368 $\\
$f( 7 )= 0.5789474 $\\
$\bar f( 5 7 )= 0.8947368 $\\
$f( 8 )= 0.5789474 $\\
$\bar f( 5 8 )= 0.8947368 $\\
\hline  \end{tabular} \end{subtable}
\begin{subtable}{.2\linewidth} \centering \begin{tabular}{|l|} \hline
\multicolumn{1}{|c|}{$X_3 = \{ 2 \},   Y_3 = \{ 5 \}$} \\ \hline
$f( 2 1 )= 0.8947368 $\\
$\bar f( 5 1 )= 0.8947368 $\\
$f( 2 3 )= 0.8947368 $\\
$\bar f( 5 3 )= 0.8947368 $\\
$f( 2 4 )= 0.8947368 $\\
$\bar f( 5 4 )= 0.8947368 $\\
$f( 2 6 )= 0.8947368 $\\
$\bar f( 5 6 )= 0.8947368 $\\
$f( 2 7 )= 0.8947368 $\\
$\bar f( 5 7 )= 0.8947368 $\\
$f( 2 8 )= 0.8947368 $\\
$\bar f( 5 8 )= 0.8947368 $\\
\hline  \end{tabular} \end{subtable} 
\begin{subtable}{.2\linewidth} \centering \begin{tabular}{|l|} \hline
\multicolumn{1}{|c|}{$X_3 = \{ 2,3\},   Y_3 = \{ 5 \}$} \\ \hline
$f( 2 3 1 )= 1 $\\
$\bar f( 5 1 )= 0.8947368 $\\
$f( 2 3 4 )= 1 $\\
$\bar f( 5 4 )= 0.8947368 $\\
$f( 2 3 6 )= 1 $\\
$\bar f( 5 6 )= 0.8947368 $\\
$f( 2 3 7 )= 1 $\\
$\bar f( 5 7 )= 0.8947368 $\\
$f( 2 3 8 )= 1 $\\
$\bar f( 5 8 )= 0.8947368 $\\
\hline  \end{tabular} \end{subtable}  \begin{subtable}{.2\linewidth} \centering \begin{tabular}{|l|} \hline
\multicolumn{1}{|c|}{$X_3 = \{ 2 \},   Y_3 = \{ 5,7 \}$} \\ \hline
$f( 2 1 )= 0.8947368 $\\
$\bar f( 5 7 1 )= 1 $\\
$f( 2 3 )= 0.8947368 $\\
$\bar f( 5 7 3 )= 1 $\\
$f( 2 4 )= 0.8947368 $\\
$\bar f( 5 7 4 )= 1 $\\
$f( 2 6 )= 0.8947368 $\\
$\bar f( 5 7 6 )= 1 $\\
$f( 2 8 )= 0.8947368 $\\
$\bar f( 5 7 8 )= 1 $\\
\hline  \end{tabular} \end{subtable} 
\begin{subtable}{.2\linewidth} \centering \begin{tabular}{|l|} \hline
\multicolumn{1}{|c|}{$X_3 = \{ \},   Y_3 = \{ 5,6 \}$} \\ \hline
$f( 1 )= 0.5789474 $\\
$\bar f( 5 6 1 )= 1 $\\
$f( 2 )= 0.5789474 $\\
$\bar f( 5 6 2 )= 1 $\\
$f( 3 )= 0.5789474 $\\
$\bar f( 5 6 3 )= 1 $\\
$f( 4 )= 0.5789474 $\\
$\bar f( 5 6 4 )= 1 $\\
$f( 7 )= 0.5789474 $\\
$\bar f( 5 6 7 )= 1 $\\
$f( 8 )= 0.5789474 $\\
$\bar f( 5 6 8 )= 1 $\\
\hline  \end{tabular} \end{subtable}  \begin{subtable}{.32\linewidth} \centering \begin{tabular}{|l|} \hline
\multicolumn{1}{|c|}{$X_3 = \{ 3 \},   Y_3 = \{6,7\}$} \\ \hline
$f( 3 1 )= 0.8947368 $\\
$\bar f( 5 6 1 )= 1 $\\
$f( 3 2 )= 0.8947368 $\\
$\bar f( 5 6 2 )= 1 $\\
$f( 3 4 )= 0.8947368 $\\
$\bar f( 5 6 4 )= 1 $\\
$f( 3 7 )= 0.8947368 $\\
$\bar f( 5 6 7 )= 1 $\\
$f( 3 8 )= 0.8947368 $\\
$\bar f( 5 6 8 )= 1 $\\
\hline  \end{tabular} \end{subtable}
\begin{subtable}{.2\linewidth} \centering \begin{tabular}{|l|} \hline
\multicolumn{1}{|c|}{$X_3 = \{ \},   Y_3 = \{ 5,6,7 \}$} \\ \hline
$f( 1 )= 0.5789474 $\\
$\bar f( 5 6 7 1 )= 1 $\\
$f( 2 )= 0.5789474 $\\
$\bar f( 5 6 7 2 )= 1 $\\
$f( 3 )= 0.5789474 $\\
$\bar f( 5 6 7 3 )= 1 $\\
$f( 4 )= 0.5789474 $\\
$\bar f( 5 6 7 4 )= 1 $\\
$f( 8 )= 0.5789474 $\\
$\bar f( 5 6 7 8 )= 1 $\\
\hline  \end{tabular} \end{subtable} 

	}
\end{table}

\end{subappendices}
\end{appendices}

\end{document}